\documentclass[sigconf, nonacm]{acmart}

\usepackage[a-1b]{pdfx}

\usepackage{bm} 
\usepackage[pdf]{graphviz}
\usepackage{dot2texi}
\makeatletter
\@ifundefined{verbatim@out}{\newwrite\verbatim@out}{}
\makeatother
\usepackage{amsfonts}
\usepackage{booktabs}
\usepackage{dsfont}
\usepackage{siunitx}
\usepackage{multirow}
\usepackage{graphicx}
\usepackage{listings}
\usepackage{commath}
\usepackage{epstopdf}
\usepackage{xcolor}
\usepackage{url}
\usepackage{mathtools}
\usepackage{caption}
\usepackage{alltt}
\usepackage{mathrsfs}
\usepackage{float}
\usepackage{subcaption}
\usepackage{graphicx,fancyvrb}

\usepackage{nicematrix}

\usepackage[linesnumbered,ruled,vlined]{algorithm2e}
\usepackage{algpseudocode}

\usepackage{amsmath,amssymb}

\usepackage{tabularx}

\usepackage{booktabs}
\usepackage{caption}
\usepackage{float}
\usepackage{capt-of}
\usepackage{booktabs}
\usepackage{colortbl}
\usepackage{xcolor}
\usepackage{xfrac}
\usepackage{footnote}
\usepackage{array}
\usepackage{arydshln}
\usepackage{epstopdf}
\DeclareGraphicsExtensions{.png,.pdf}
\usepackage{amsmath,amsfonts,amssymb}

\definecolor{ochre}{HTML}{C06000}

\usepackage{array}

\DeclarePairedDelimiterX{\infdivx}[2]{(}{)}{%
	#1\;\delimsize|\delimsize|\;#2%
}

\definecolor{mygreen}{rgb}{0,0.6,0}
\definecolor{mygray}{rgb}{0.5,0.5,0.5}
\definecolor{mymauve}{rgb}{0.58,0,0.82}

\usepackage{listings}
\lstnewenvironment{VerbatimText}[1][]{
    
    \lstset{fancyvrb=true,basicstyle=\footnotesize,captionpos=b,xleftmargin=2em,#1}
}{}

\usepackage{booktabs}
\usepackage{pgfplots}
\usepackage{pgfplotstable}

\usepackage{amsmath}
\DeclareMathOperator*{\argmax}{argmax}

\newcommand{\ignore}[1]{}

\newcommand{\defeq}{\stackrel{\text{def}}{=}}
\newcommand*{\rom}[1]{\expandafter\@slowromancap\romannumeral #1@}

\usepackage{esvect}

\newcommand{\vo}{\vec{o}\@ifnextchar{^}{\,}{}}

\colorlet{lightgray}{gray!20}

\newcommand{\RNum}[1]{\uppercase\expandafter{\romannumeral #1\relax}}

\newcommand{\proj}[1]{{\Pi}}
\newcommand{\sel}[1]{{\sigma}}

\newcommand{\cut}[1]{}
\newcommand{\eat}[1]{}

\newtheorem{theorem}{Theorem}
\newtheorem{lemma}[theorem]{Lemma}
\newtheorem{corollary}[theorem]{Corollary}
\newtheorem{definition}[theorem]{Definition}

\newcommand{\setof}[2]{\{{#1}\mid{#2}\}}        

\newcommand{\R}{\mathbb R} 
\newcommand{\extend}[1]{\bm #1}

\newcommand{\opt}{{\mathcal{OPT}}}

\usepackage[subtle]{savetrees}

\SetKwInput{KwParameter}{Parameters}
\SetKwRepeat{Do}{do}{while}

\begin{document}
\title{Quasi-stable Coloring for Graph Compression}
\subtitle{Approximating Max-Flow, Linear Programs, and Centrality}

\newcommand\vldbavailabilityurl{https://github.com/mkyl/QuasiStableColors.jl}

\author{Moe Kayali}
\orcid{0000-0002-0643-6468}
\affiliation{
    \institution{University of Washington}
   \country{}
    }
\email{kayali@cs.washington.edu}

\author{Dan Suciu}
\orcid{0000-0002-4144-0868}
\affiliation{
    \institution{University of Washington}
    \country{}
    }
\email{suciu@cs.washington.edu}

\makeatletter
\let\@authorsaddresses\@empty
\makeatother

\begin{abstract}
  We propose \textit{quasi-stable coloring}, an approximate version of
  stable coloring. Stable coloring, also called color refinement, is a
  well-studied technique in graph theory for classifying vertices,
   which can be used to build
  compact, lossless representations of graphs. However, its usefulness
  is limited due to its reliance on strict symmetries.  Real data
  compresses very poorly using color refinement. We propose the first,
  to our knowledge, approximate color refinement scheme, which we call
  quasi-stable coloring. By using approximation, we alleviate the need
  for strict symmetry, and allow for a tradeoff between the degree of
  compression and the accuracy of the representation.  We study three
  applications: Linear Programming, Max-Flow, and Betweenness
  Centrality, and provide theoretical evidence in each case that a
  quasi-stable coloring can lead to good approximations on the reduced
  graph.  Next, we consider how to compute a maximal quasi-stable coloring: we
  prove that, in general, this problem is NP-hard, and propose a
  simple, yet effective algorithm based on heuristics.  Finally, we
  evaluate experimentally the quasi-stable coloring technique on
  several real graphs and applications, comparing with prior approximation techniques.
\end{abstract}

\maketitle

\ifdefempty{\vldbavailabilityurl}{}{
\begingroup\small\noindent\raggedright\textbf{Artifact Availability:}\\
The source code, data, and/or other artifacts have been made available at \url{\vldbavailabilityurl}.
\endgroup
}

\setcounter{page}{1}

\section{Introduction}
\label{sec:intro}
A well known technique for finding structure in large graphs is the
{\em color refinement}, or {\em 1-dimensional Weisfeiler-Leman}
method.  It consists of assigning the nodes in the graph some initial
color, for example based on their labels.  Then one repeatedly refines the
coloring, by assigning distinct colors to two nodes whenever those
nodes have a different number of neighbors of the same color; when no
more refinement is possible, then this is called a {\em stable
  coloring}.  We show a simple illustration in Fig.~\ref{fig:example}
(a): for example nodes 5 and 11 have the same dark purple color, because both have
one purple, one lavender, and one dark purple neighbor, while node 7 has a different
color because it additionally has an olive neighbor.  The stable coloring
can be computed efficiently, in almost linear
time~\cite{DBLP:journals/siamcomp/PaigeT87}, can be generalized to
labeled graphs, weighted graphs, directed or undirected graphs, and
multigraphs, and is used by graph isomorphism
algorithms~\cite{10.1145/3372123}, in graph
kernels~\cite{DBLP:conf/nips/ShervashidzeB09}, and more recently
has been used to explain and enhance the power of Graph Neural
Networks~\cite{DBLP:conf/iclr/XuHLJ19,DBLP:conf/aaai/0001RFHLRG19,DBLP:conf/pods/Grohe20,DBLP:conf/ijcai/0001FK21}.
We review stable coloring in Sec.~\ref{sec:background}.  Excellent
surveys of color refinement and its recent applications to machine
learning can be found
in~\cite{DBLP:journals/corr/abs-2112-09992,DBLP:conf/pods/Grohe20,DBLP:conf/lics/Grohe21}.

\begin{figure}
  \begin{subfigure}[t]{\linewidth}
    \includegraphics[width=\linewidth]{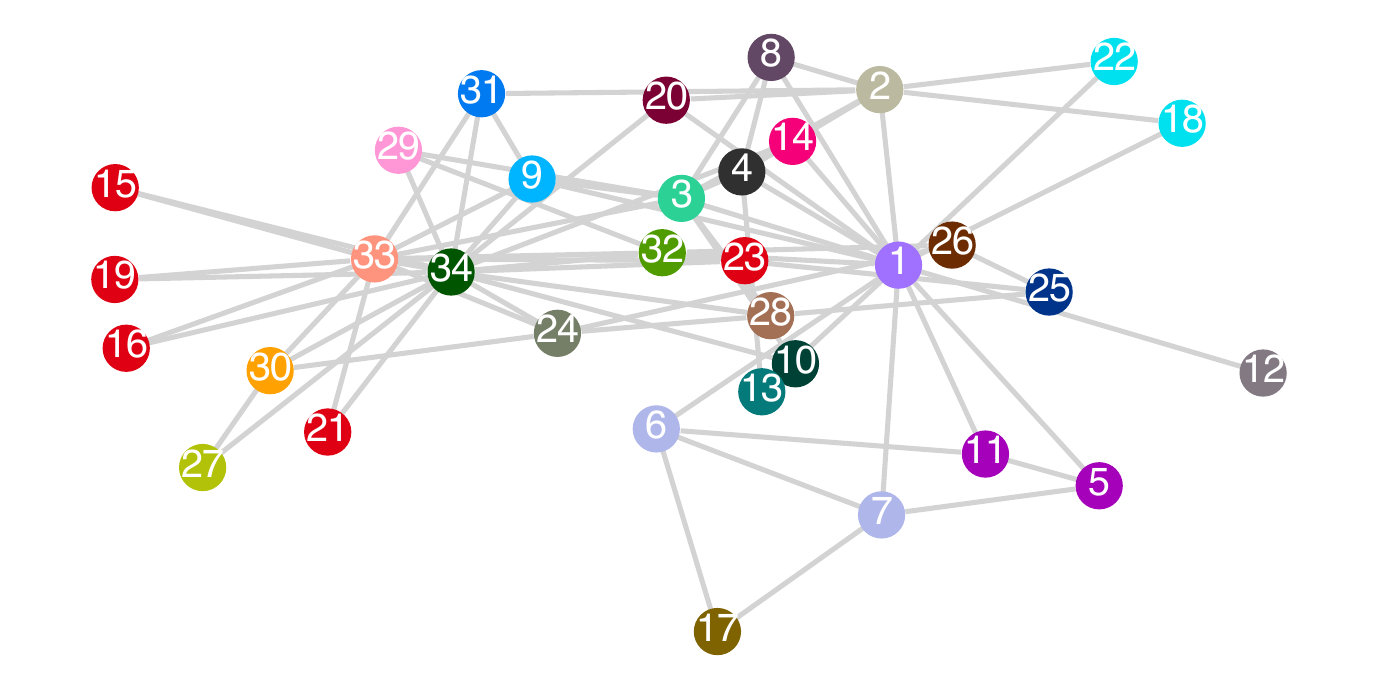}
    \caption{Stable coloring}
  \end{subfigure}
  \begin{subfigure}[b]{\linewidth}
    \includegraphics[width=\linewidth]{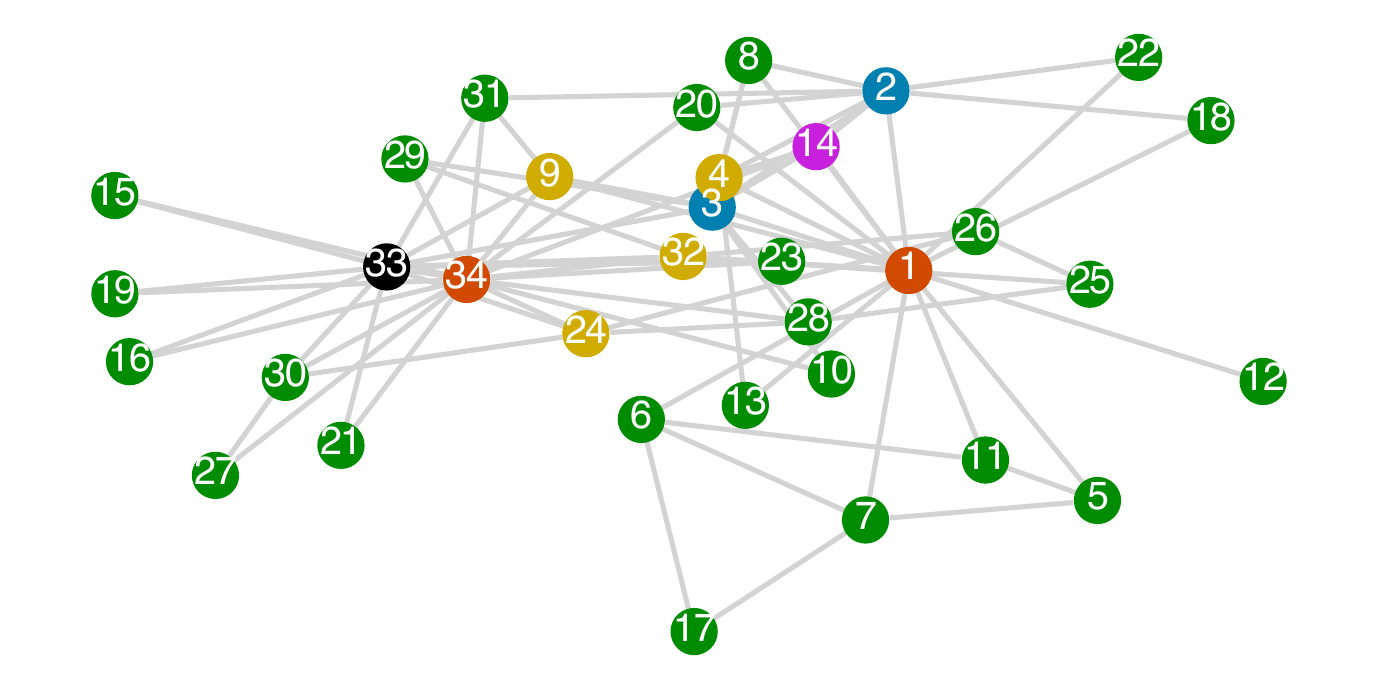}
    \caption{Quasi-stable coloring, $q=3$}
    \label{fig:karate-q-stable}
  \end{subfigure}
  \caption{Coloring Zachary's karate club
  graph~\cite{zachary1977information} where
   $|V|=34$, $|E|=78$. While the stable coloring requires $27$ colors, 
  for a quasi-stable color $6$ colors suffice when
  $q=3$. Note in~\ref{fig:karate-q-stable} the club leaders $\set{1, 34}$ are put into their
  own color.}
  \label{fig:example}
\end{figure}

In this paper we apply stable coloring as a compression technique for
large graphs.  A stable coloring naturally defines a reduced graph,
whose nodes are the classes of colors of the original graph.  The new
graph preserves many important properties of the original graph, which
makes it a good candidate for compression.  For example, an elegant
theoretical result states that the reduced graph satisfies
precisely the same properties expressible in the $C^2$ logic as the
original graph~\cite{grohe2017descriptive}.  Motivated by the fact
that the reduced graph preserves key properties of the original graph,
Grohe et al.~\cite{grohe2014dimension,Grohe:2021aa} propose using
color refinement as a dimensionality reduction technique, and show
two applications: to Linear Programming and to graph kernels.

However, we notice that stable coloring is not effective for
dimensionality reduction because, in practice, the ``reduced'' graph
is only slightly smaller than the original graph.  For example, the
graph in Fig~\ref{fig:example} (a) has 34 nodes, but 27 colors, which
means that the reduced graph has 80\% of the number of nodes of the
original.  We show in Sec.~\ref{sec:evaluation} that, for typical
large graphs, the size of the reduced graph is between 70\% - 80\%
that of the original graph.  Moreover, even when a graph happens to have a small
reduced graph, any tiny update, e.g. adding or deleting an edge, will
immediately lead to a huge increase of the reduced graph.  We
illustrate this phenomenon briefly in
Fig.~\ref{fig:graphs-stable-coloring}: we started with an
artificially regular graph, which compressed well from 1000 nodes to
only 100 colors, but the compression degrades very rapidly when we add
only a few edges.

\begin{figure}
  \centering  
    \includegraphics[width=\linewidth]{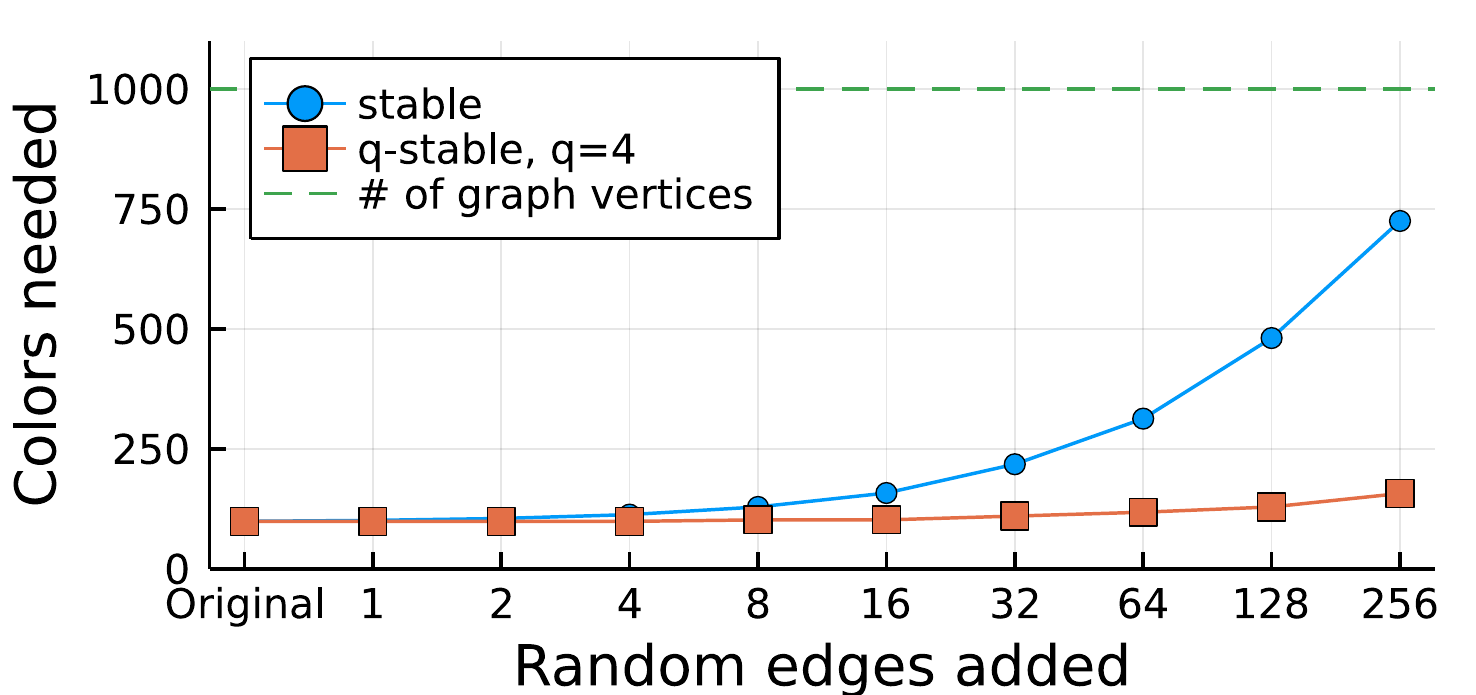}
    \caption{Comparing the robustness of stable and q-stable coloring. A
    synthetic graph with $|V| = 1~000, |E|=21~600$ is generated with a size $100$ 
    stable and q-stable coloring. As a small fraction of edges (no
    more than $1.5\%$) are
    added, the brittleness of stable coloring causes compression to
    degenerate---quasi-stable colors are immune to this.}
  \label{fig:graphs-stable-coloring}
\end{figure}

In this paper we propose a generalization of stable coloring
to {\em quasi-stable coloring}, with a goal of reducing the size of
the compressed graph while still allowing applications to be
processed approximatively.  Our definition of $q$-stable coloring,
given in Sec.~\ref{sec:approximate}, allows two nodes to be in the
same color if they have a number of neighbors to any another color
that differs by at most $q$, where $q \geq 0$ is a parameter.  This has
dual effect.  First, it can dramatically reduce the size of the
compressed graph, since more nodes can now be assigned the same color.
For example, the quasi-stable coloring in Fig.~\ref{fig:example} (b)
allows nodes 5 and 7 to have the same color, even though their number
of green neighbors differs by 1; the new compressed graph has only 6
nodes. Second, this makes the technique less sensitive to
data updates, because it can tolerate nodes with a slightly
different number of neighbors.  We can see in
Fig.~\ref{fig:graphs-stable-coloring} that the number of
quasi-stable colors increases only marginally with the addition of random
edges.  By varying the parameter $q$, the quasi-stable colors offer
a tradeoff between the compression ratio and the degree to which the
reduced graph preserves the properties of the original graph.

We start by investigating in Sec.~\ref{sec:applications} whether the
result of an application over the reduced graph is a good
approximation of the result on the original graph.  We consider three
applications: linear programming, max flow, and betweenness
centrality.  In all three cases we provide a theoretical justification
for why the result on the reduced graph should be close to the true
value.  First, for linear programs, we prove that the optimal value of
the reduced program converges to the optimal value of the original
program when $q\rightarrow 0$.  This generalizes the result
in~\cite{grohe2014dimension}, which proved that, when the coloring is
stable ($q=0$) then the LP has the same optimal value on reduced
program and the original program.  Next, for the max flow problem we
prove that, while pathological cases exist where a ``good''
quasi-stable coloring has a totally different max-flow than the
original graph, under reasonable assumptions the two are close and, in
particular, when the coloring is stable ($q=0$) then they are
equal. Third, we examine betweenness centrality, and show that, even a
stable coloring can, in pathological cases, lead to different
centrality scores, but we prove that the 2-WL method (a refinement of
1-WL) always preserves the centrality score.

Next, in Section~\ref{sec:algo} we study the algorithmic problem of
efficiently computing a quasi-stable coloring for a graph.  While the
stable coloring can be computed in almost linear time, we prove,
rather surprisingly, that finding an optimal quasi-stable coloring is
NP-hard.  The main difference is that there is always a maximum,
``best'', stable coloring, but none exist in general for quasi-stable
coloring.  Based on this observation, we propose a heuristic-based
algorithm for computing a quasi-stable coloring, whose decisions are
informed by our theoretical analysis in Sec.~\ref{sec:applications}.

Finally, we conduct an empirical evaluation of quasi-stable coloring
in Section~\ref{sec:evaluation}. Testing on twenty datasets from a variety
of domains, we find q-stable colorings favorably trade-off accuracy
for speed. For example, on the \texttt{qap15} linear program, an exact
solution takes 22~minutes to compute while the q-stable approximation 
reduces the problem size by $100\times$, solving it end-to-end within
17~seconds while introducing
only a 5\% error. We observe similar trends for max-flow and
centrality applications. Next, we study the characteristics of the
compressed graphs, finding that they avoid the pitfalls of stable
colorings. We conclude by characterizing our algorithm, analyzing its
runtime, its ability to progressively improve the approximations and
 testing its robustness to noise.

In summary, we make the following contributions in this paper:
\begin{itemize}
\item We propose a relaxation of stable coloring, called quasi-stable
  coloring; Sec.~\ref{sec:approximate}.
\item We provide theoretical evidence that the reduced graph defined
  by a quasi-stable coloring can be useful in three classes of
  applications; Sec.~\ref{sec:applications}.
\item We prove that an optimal quasi-stable coloring is NP-hard to
  compute, and propose an efficient, heuristic-based algorithm;
  Sec.~\ref{sec:algo}.
\item We conduct an experimental evaluation on several real graphs and
  applications; Sec.~\ref{sec:evaluation}.
\end{itemize}
\vspace{-1em}
\section{Background on Color Refinement}

\label{sec:background}

Fix an undirected graph $G = (V,E)$.  We denote by $N(x)$ the set of
neighbors of a node $x \in V$.  A {\em coloring} of $G$ is a partition
of $V$ into $k$ disjoint sets, $V = P_1 \cup \cdots \cup P_k$. We say
that a node $x \in P_i$ has color $i$, or that it has color $P_i$.  We
denote the coloring by $\bm P = \set{P_1, \ldots, P_k}$.  A {\em
  stable coloring} is a coloring with the property that, for any two
colors, all nodes in the first color have the same number of neighbors
in the second color.  Formally:
\vspace{-0.275em}
\[ \forall i, j, \forall x, y \in P_i: \ \ |N(x) \cap P_j|=|N(y) \cap P_j| \]
Given two colorings $\bm P, \bm P'$ we say that the first is a
refinement of the second, and denote this by $\bm P \subseteq \bm P'$,
if for every color $P_i \in \bm P$, there exists a color
$P_j' \in \bm P'$ such that $P_i \subseteq P_j'$.  Any two colorings
$\bm P, \bm P'$ have a greatest lower bound, $\bm P \wedge \bm P'$,
and a least upper bound, $\bm P \vee \bm P'$.  The greatest lower
bound is easily constructed, by considering the partition
$\setof{P_i \cap P_j'}{P_i \in \bm P, P_j' \in \bm P'}$; for a
construction of $\bm P \vee \bm P'$, we refer the reader
to~\cite{MR2868112}.

The smallest coloring, where each node $x$ is in a separate color,
denoted $\bm P_\bot$, is trivially a stable coloring.  Somewhat less
obvious is the fact that, if both $\bm P$ and $\bm P'$ are stable
colorings, then their least upper bound $\bm P \vee \bm P'$ is also a
stable coloring (see also Th.~\ref{th:lub} below).  This implies that
every graph has a unique, maximum stable coloring, often called
\underline{\em the} stable coloring of $G$, namely
$\bm P_1 \vee \bm P_2 \vee \cdots$, where $\bm P_1, \bm P_2, \ldots$
are all stable colorings of the graph.  The stable coloring can be
computed quite efficiently using the {\em color refinement method},
sometimes also called the {\em 1-dimensional Weifeiler-Leman method}
(1WL).  Start by coloring all nodes with the same color, then
repeatedly choose two colors $P_i, P_j$ and refine the set $P_i$ by
partitioning its nodes based on their number of neighbors in $P_j$;
the stable coloring is obtained when no more refinement is possible.
There exist improved algorithms that compute the stable coloring in
time $O(n + m \log n)$, where $n,m$ are the number of nodes and edges
respectively~\cite{DBLP:journals/corr/abs-2112-09992}.  The {\em
  reduced graph}, $\hat G$, has one node $i$ for each color $P_i$, and
an edge from $i$ to $j$ if some node $x \in P_i$ has a neighbor
$y \in P_j$ (in which case {\em every} node $x \in P_j$ has a neighbor
$y \in P_j$).

Color refinement can be generalized to directed graphs, to labeled
graphs, to multi-graphs, and to weighted graphs.  We refer the reader
to~\cite{DBLP:journals/corr/abs-2112-09992} for an extensive survey of
the theoretical properties of the color refinement method.  In this
paper we will consider directed, weighted graphs, but defer their
discussion to Sec.~\ref{sec:approximate}.

Most applications of stable coloring work best on graphs that have
{\em many} colors, i.e. where there are many, small sets $P_i$.  For
example, in order to check for an isomorphism between two graphs
$G, G'$, one first computes the stable coloring of the disjoint union
of $G$ and $G'$, then restricts the isomorphisms candidates to
functions that preserve the color of the nodes.  The best case is when
the stable coloring is $\bm P_\bot$, because then the only possible
isomorphism is the function that maps $x\in G$ to the similarly
colored $y \in G'$.  In general, applications of the 1WL method work
best when there are many colors.

In this paper we use coloring for dimensionality reduction and
approximate query processing.  Instead of solving the problem on the
original, large graph $G$, we solve it on the reduced graph $\hat G$.
This technique works best when there are {\em few} colors, because
then the reduced graph is small. Real graphs tend to have many colors; we found (see
Sec.~\ref{sec:evaluation}) that the number of colors is typically
around 70\% of the number of nodes.  This observation has a theoretical
justification: if $G$ is a random graph, then with high probability
its stable coloring is
$\bm P_\bot$~\cite[Sec.3.3]{DBLP:journals/corr/abs-2112-09992}.  This motivated
us to introduce a new notion, {\em quasi-stable coloring}, which
relaxes the stability condition, in order to allow us to construct
fewer, larger colors.

\section{Quasi-stable Coloring}

\label{sec:approximate}

We have seen that the stable coloring of a graph has many elegant
properties, but offers poor compression in practice.  In this section
we introduce a relaxed notion, which preserves some of the desired
properties while improving the compression.

A {\em weighted} directed graph $G=(X,E,w)$ is a directed graph with a
function $w$ mapping edges to real numbers.  We will assume that the
edge $(x,y)$ exists iff $w(x,y)\neq 0$, and therefore we often omit
$E$ and simply write the directed graph as $G=(X,w)$.  Conversely,
given a standard graph $G=(X,E)$, we assume a default weight function
$w(x,y)=1$ when $(x,y) \in E$ and $w(x,y)=0$ otherwise.  A {\em
  bipartite graph} is a graph where the nodes consist of two sets
$X,Y$ and all edges go from some node in $X$ to some node in $Y$.  We
denote a bipartite graph by $(X,Y,E)$ or $(X,Y,w)$ if it is weighted.

Given a weighted graph $(X,w)$ and two subsets of nodes $U, V$, we
denote by $w(U,V)$ the total weight from $U$ to $V$:
\begin{align}
  w(U,V) \defeq & \sum_{x \in U, y \in V} w(x,y)
\label{eq:plus:minus}
\end{align}

Fix a reflexive and symmetric relation $\sim$ on $\R$.

\begin{definition} \label{def:quasi:stable} (1) Let $G = (X,Y, w)$ be
  a weighted, bipartite graph (i.e. $w : X \times Y \rightarrow
  \R$). We say that $G$ is {\em $\sim$regular} if the following two
  conditions hold:
  \begin{align*}
    \forall x_1, x_2 \in X: \ w(x_1,Y)\sim & w(x_2,Y) \\
    \forall y_1, y_2 \in Y: \ w(X,y_1)\sim & w(X,y_2)
  \end{align*}
  (2) Let $G=(X,w)$ be any weighted, directed graph, and
  $\bm P= \set{P_1, \ldots, P_k}$ be a partition of its nodes.  We say
  that $\bm P$ is {\em $\sim$quasi-stable}, or {\em quasi-stable
    w.r.t. $\sim$}, if, for any two colors $P_i, P_j$ (including
  $i=j$), the bipartite graph $(P_i, P_j, w)$ is $\sim$regular.
\end{definition}

Thus, a quasi-stable coloring partitions the nodes in such a way that
for any two colors $P_i, P_j$, any two nodes in $P_i$ have similar
(according to $\sim$) outgoing weights to $P_j$, and any two nodes in
$P_j$ have similar incoming weights from $P_i$.

\subsection{Examples}
\label{sec:coloring-types}

We illustrate with several examples.

{\bf Biregular Graphs, and Stable Coloring} Recall that a bipartite
graph $(X,Y,E)$ is $(a,b)$-biregular, or simply biregular when $a,b$
are clear from the context, if every node $x \in X$ has outdegree $a$
and every node in $y \in Y$ has indegree $b$.  Let $\sim$ be the
equality relation on $\R$: $u \sim v$ iff $u = v$.  Then $(X,Y,E)$ is
$=$regular iff it is biregular.  Furthermore, if $G$ is a directed
graph, then a coloring $\bm P$ is $=$quasi-stable iff it is stable.

{\bf $q$-Stable Coloring} The main type of quasi-stable coloring that
we use in this paper is called $q$-stable.  Fix some number
$q \geq 0$, and define the following similarity relation on $\R$:
$u \sim_q v$ if $|u - v| \leq q$.  Notice that $\sim_q$ is reflexive
and symmetric, but not transitive.  In a $\sim_q$regular
bipartite graph any two nodes in $X$ have outgoing weights that differ
by at most $q$, and similarly for the incoming weights of the nodes in
$Y$.  To reduce clutter, we will call a $\sim_q$quasi-stable coloring
 a {\em $q$-stable coloring}, or simply a {\em $q$-coloring}.
The standard stable coloring is the special case when $q=0$.

{\bf $\varepsilon$-Relative Coloring} While $q$-stable coloring
imposes a bound on the absolute error, we briefly discuss an
alternative: imposing a bound on the relative error.  Fix some number
$\varepsilon \geq 0$, and define $u \sim^\varepsilon v$ as
$u \cdot e^{-\varepsilon} \leq v\leq u \cdot e^{\varepsilon}$.  This
relation is reflexive and symmetric, but not transitive.  We call a
$\sim^\varepsilon$quasi-stable coloring simply a {\em
  $\varepsilon$-relative coloring}.  Notice that isolated nodes
(i.e. without incoming or outgoing edges) are in a separate color.
This is because zero is similar only to itself: $u \sim^\varepsilon 0$
implies $u=0$.  More generally, for any two colors, either every node
in the first color is connected to some node in the second, or none
is, and similarly with the role of the two colors reversed.

{\bf Bisimulation Relation} As a last example, define $u \equiv v$ as
$u=v=0$ or $u\neq 0, v\neq 0$.  In other words, $\equiv$ checks if
both $u, v$ are zero, or none is zero.  This is an equivalence
relation.  Then, a $\equiv$quasi-stable coloring is a {\em
  bisimulation} relation on that graph~\cite{grohe2017descriptive}.

\subsection{The Reduced Graph}

Let $G=(X,w)$ be a directed, weighted graph and let
$\bm P = \set{P_1, \ldots, P_k}$ be any coloring, not necessarily
quasi-stable.  The {\em reduced graph}, is defined as
$\hat G \defeq (\hat X, \hat w)$, where the nodes
$\hat X \defeq \set{1,2,\ldots, k}$ correspond to the colors.  We will
consider different choices for the weight function; one example is
that we can set it to be the sum of all weights between two colors,
i.e.  $\hat w(i, j) \defeq \sum_{x \in P_i; y \in P_j} w(x,y)$, but we
will consider other options too.  Our goal in this paper is to use the
reduced graph to compute approximate answers to problems that are
expensive to compute on the original, large graph.

\section{Applications}
\label{sec:applications}

Stable coloring preserves many nice properties of the graph.  Will a
$q$-quasi stable coloring preserve such properties to some degree?  We
explore this question here, and provide theoretical evidence that
quasi stable colorings provide some useful approximations for three
problems: linear optimization, maximum flow, and betweenness
centrality.  In Section~\ref{sec:evaluation} we validate
experimentally these findings.

\subsection{Linear Optimization}

\label{sec:linear:optimization}

We start with Linear Optimization.  Consider the following linear
program:
\begin{align}
  \text{maximize } & c^T x  & \text{where } & A x \leq b \text{ and } x \geq 0 \label{eq:sys:1}
\end{align}
where $A \in \R^{m \times n}$, $b \in \R^m$, $c \in R^n$.  We will
denote by $\opt(A, b, c)$ the optimal value of $c^T x$.  In general,
it is possible to have $\opt=-\infty$ (namely when the set of
constraints is infeasible), or $\opt=\infty$, but we will not consider
these cases.  We will apply quasi-stable coloring only to LPs that are
{\em well behaved}, meaning that $\opt(A,b,c)$ is finite, and
continues to be finite when $b,c$ range over some small neighborhood.

In this section, we view a matrix $A$ as a function $A(i,j)$.
Following the notation in~\eqref{eq:plus:minus}, we denote
$A(P,Q) \defeq \sum_{i \in P, j \in Q} A(i,j)$ when
$P \subseteq [m], Q \subseteq [n]$, and similarly,
$b(P) \defeq \sum_{i \in P} b(i)$, $c(Q)=\sum_{j \in Q}c(j)$.  We
write {\em boldface $\bm A$} for the extended matrix of the LP:
\begin{align}
  \extend{A} \defeq & \left(
                      \begin{array}{ c |c}
                      A & b \\ \hline
                      c^T & \infty
                    \end{array}
\right) \label{eq:a:ext}
\end{align}
The last row, $m+1$, is the vector $(c^T, \infty)$, and the last
column, $n+1$, is the vector $(b, \infty)$.  We associate the LP with
the weighted bipartite graph $G = ([m+1],[n+1],\extend{A})$, where the
weights are the matrix entries (they may be $<0$).

Consider a coloring $(\bm P, \bm Q)$ of the bipartite graph $G$; it
partitions the $[m+1]$ rows into $P_1, \ldots, P_k, P_{k+1}$, and the
$[n+1]$ columns into $Q_1, \ldots, Q_\ell, Q_{\ell+1}$.  We further
assume that the last row and last column of $\extend{A}$ have a unique
color, namely $P_{k+1} = \set{m+1}$, and $Q_{\ell+1} = \set{n+1}$.
The partition defines a reduced bipartite graph,
$\hat G = ([k+1], [\ell+1], \hat{\extend{A}})$, where we define the
weights as follows:
\begin{align}
  \hat{\extend{A}}(r,s) \defeq & \frac{\extend{A}(P_r,Q_s)}{\sqrt{|P_r|\cdot |Q_s|}}\label{eq:hat:a:def}
\end{align}
In other words, the weight of the edge from color $r$ to color $s$ is
the sum of all $\extend{A}_{ij}$ with $i$ in color $P_r$ and $j$ in
color $Q_s$, normalized by $\sqrt{|P_r|\cdot|Q_s|}$.  The {\em reduced
  LP} is the LP defined by the matrix $\hat{\extend{A}}$ of the
reduced graph.  In other words, the reduced LP is the following:
\begin{align}
  \text{maximize } & \hat c^T \hat x &  \text{where } & \hat A \hat x \leq \hat b  \text{ and } \hat x \geq 0 \label{eq:sys:2} 
\end{align}
where $\hat A, \hat b, \hat c$ are defined as:
\begin{align}
\hat A(r,s) \defeq & \frac{A(P_r,Q_s)}{\sqrt{|P_r|\cdot |Q_s|}} & \hat b(r) \defeq &  \frac{b(P_r)}{\sqrt{|P_r|}} & \hat c(s) \defeq & \frac{c(Q_s)}{\sqrt{|Q_s|}}
\label{eq:def:hat:a:b:c}
\end{align}

We prove that, if the coloring is quasi-stable, then the solution to
problem~\eqref{eq:sys:1} is close to that of problem~\eqref{eq:sys:2}.

\begin{theorem} \label{th:approx:lp} Assume that the LP defined by
  $A, b, c$ is well behaved.  Then there exists $q_0>0$ that depends
  only on $A, b, c$, such that, for all $q \leq q_0$, for any
  $q$-quasi stable coloring,
  $|\opt(A,b,c)-\opt(\hat A,\hat b, \hat c)| \leq q \Delta$, where
  $\hat A, \hat b, \hat c$ is the reduced LP associated to the
  coloring, and the constant $\Delta$ depends only on $A,b,c$.
\end{theorem}

\newcommand\y{\cellcolor{green!10}}
\newcommand\gr{\cellcolor{yellow!10}}
\newcommand\bl{\cellcolor{blue!10}}
\newcommand\rd{\cellcolor{red!10}}
\newcommand\pl{\cellcolor{cyan!10}}
\newcommand\gy{\cellcolor{black!10}}
\newcommand\pk{\cellcolor{orange!10}}
\newcommand\ma{\cellcolor{magenta!10}}

\begin{figure*}
  \begin{minipage}[c]{0.23\linewidth}
    \begin{align*}
      \text{maximize } & 9x_1+10x_2+50x_3 \\
      \text{where } & 4x_1+8x_2+2x_3 \leq 20 \\
                       & 6x_1 + 5x_2 + x_3 \leq 20 \\
                       & 7x_1 + 4x_2 + 2x_3 \leq 21 \\
                       & 3x_1 + x_2 + 22x_3 \leq 50 \\
                       & 2x_1 + 3x_2 + 21x_3 \leq 51
    \end{align*}
    \centerline{Optimal value: 128.157}
    \centerline{(a)}
  \end{minipage}
  \begin{minipage}[c]{0.45\linewidth}
    \begin{align*}
      \extend{A} =
      &
        \left(
        \begin{array}{cc|c|c}
          \y 4 & \y 8 & \gr 2 & \bl 20 \\
          \y 6 & \y 5 & \gr 1 & \bl 20 \\
          \y 7 & \y 4 & \gr 2 & \bl 21 \\
          \hline
          \rd 3 & \rd 1 & \pl 22 & \gy 50 \\
          \rd 2 & \rd 3 & \pl 21 & \gy 51 \\
          \hline
          \pk 9 & \pk 10 & \ma 50 & \infty
        \end{array}
                        \right)
                    & \hat{\extend{A}} =
            & \left(
              \begin{array}{c c|c}
                \y \frac{34}{\sqrt{3\cdot 2}}& \gr \frac{5}{\sqrt{3\cdot 1}}&\bl \frac{61}{\sqrt{3\cdot 1}}\\
                \rd \frac{9}{\sqrt{2\cdot 2}}& \pl \frac{43}{\sqrt{2\cdot 1}}& \gy \frac{101}{\sqrt{2\cdot 1}}\\
                \hline
                \pk \frac{19}{\sqrt{1\cdot 2}}& \ma \frac{50}{\sqrt{1\cdot 1}} & \infty
              \end{array}
                                                                        \right)
    \end{align*}
    \centerline{(b)}
  \end{minipage}
  \begin{minipage}[c]{0.23\linewidth}
    \begin{align*}
      \text{maximize } & \frac{19}{\sqrt{2}}\hat x_1 + 50 \hat x_2\\
      \text{where } & \frac{34}{\sqrt{2}}x_1+5x_2 \leq 61 \\
      & \frac{9}{\sqrt{2}}+43 x_2 \leq 101
    \end{align*}
    \centerline{Optimal value: 130.199}
    \centerline{(c)}
  \end{minipage}
  \caption{Example of (a) Linear Program, (b) constraint matrix reduced
  via q-stable coloring, and (c)
    the reduced Linear Program.}
  \label{fig:example:lp}
\end{figure*}

We give the proof in Appendix~\ref{app:proof:th:approx:lp}. The
theorem guarantees that, by improving the quality of the quasi-stable
coloring, the value of the reduced linear program eventually converges
to the true value.

\begin{example}
  Consider the linear program in Fig.~\ref{fig:example:lp} (a).  Its
  matrix has dimensions $5 \times 3$.  Fig.~\ref{fig:example:lp} (b)
  shows a block-parition of the extended matrix $\extend{A}$, which
  corresponds to a $q$-quasi stable coloring, for $q=1$.  More
  precisely, in each block, the row-sums differ by at most 1, and the
  column sums differ by at most 1.  For example, the three rows in the
  first block have sums $4+8=12$ and $6+5=11$ and $7+4=11$, so they
  differ by at most 1, while the column-sums are equal.  The reduced
  matrix is shown in Fig.~\ref{fig:example:lp} (c).  The optimal value
  of the original LP is 128.157 and that of the reduced LP is 130.199.
\end{example}

{\bf Discusssion} Mladenov et
al.~\cite{DBLP:journals/jmlr/MladenovAK12} and Grohe et
al.~\cite{grohe2014dimension} used stable coloring of the matrix
(which is also called there an {\em equitable partition of
  $\extend{A}$}) to reduce the dimensionality of a linear program.  We
recover their result as the special case when $q=0$: in that case, our
theorem above implies $\opt(A,b,c)=\opt(\hat A, \hat b, \hat c)$.  The
reduced LP in~\cite{grohe2014dimension} is different from ours,
however, we explain here that both are special cases of a more general
form of reduction.  To explain that, recall that a {\em fractional
  isomorphism} from $\extend{A}$ to $\hat{\extend{A}}$ (Equations
(5.1), (5.2) in~\cite{grohe2014dimension}) is a pair of {\em
  stochastic} matrices $\extend{U}, \extend{V}$ such that:
\begin{align}
\extend{A}\extend{V}^T=&\extend{U}^T\hat{\extend{A}}&
\extend{U}\extend{A}=&\hat{\extend{A}}\extend{V} \label{eq:def:u:v:grohe}
\end{align}
Our proof in Appendix~\ref{app:proof:th:approx:lp} in fact shows the
following:
\begin{theorem}[Informal]  \label{th:linear-eq:Informal}
  If $\hat{\extend{A}}$, $\extend{U}$, $\extend{V}$ are three matrices
  such that Equations~\eqref{eq:def:u:v:grohe} hold exactly, or hold
  approximatively, and $\extend{U}$, $\extend{V}$ are non-negative,
  then $\opt(A,b,c)$ and $\opt(\hat{\extend{A}},\hat b, \hat c)$ are
  equal, or are approximatively equal.
\end{theorem}
Notice that we do not require $\extend{U},\extend{V}$ to be stochastic,
only non-negative. In fact, our particular choice
$\extend{U},\extend{V}$ in~\eqref{eq:def:u:v} are not stochastic.  By
using this result we derive many other choices for the reduced LP, as
follows.  Let $\extend{M}, \extend{N}$ be any diagonal matrices of
dimensions $(k+1)\times (k+1)$ and $(\ell+1)\times(\ell+1)$
respectively, where all elements on the diagonal are $>0$, and define:
\begin{align*}
  \hat{\extend{A}}' \defeq & \extend{M}\hat{\extend{A}}\extend{N}^{-1}
& \extend{U}' \defeq & \extend{M} \extend{U}
& \extend{V}' \defeq & \extend{N} \extend{V}
\end{align*}
The reader may check that equations~\eqref{eq:def:u:v:grohe} continue
to hold (exactly or approximatively) when we replace
$\hat{\extend{A}},\extend{U},\extend{V}$ with
$\hat{\extend{A}}',\extend{U}',\extend{V}'$.  Now we can explain the
construction of the matrix that defines the reduced LP
in~\cite{grohe2014dimension}.  Start from~\eqref{eq:hat:a:def} (or,
equivalently, from~\eqref{eq:def:hat:a:b:c}), and define the diagonal
matrices:
\begin{align*}
\extend{M} \defeq & \texttt{diag}(\sqrt{|P_1|},\ldots, \sqrt{|P_{k+1}|})
& \extend{N} \defeq & \texttt{diag}(\sqrt{|Q_1|},\ldots, \sqrt{|Q_{\ell+1}|})
\end{align*}
Then the new matrix $\hat{\extend{A}}'$ defines  reduced LP
in~\cite{grohe2014dimension}.  More precisely:
\begin{align*}
\hat{A}'(r,s) \defeq & A(P_r,Q_s)/|Q_s| & \hat{b}'(r) \defeq & b(P_r) &\hat{c}'(s) \defeq &c(Q_s)/|Q_s|
\end{align*}

\subsection{Maximum Flow}

Next, we consider the maximum flow problem.  We show that, while in
general quasi-stable coloring may not necessarily lead to a good
approximate solution, we describe a reasonable property under which it
does.  In particular, our result implies that stable coloring always
preserves the value of the maximum flow.

In the {\em network flow problem} we are given a network
$G = (X,c,S,T)$ where $X$ is a set of nodes,
$c : X \times X \rightarrow \R_+$ is a {\em capacity function}, and
$S, T \subseteq X$ are sets of nodes called source and target nodes.
A {\em flow} is a function $f:X \times X \rightarrow \R_+$ satisfying
the {\em capacity condition}, $f(x,y) \leq c(x,y)$,
$\forall x,y\in X$, and the {\em flow preservation condition},
$f(X,z)=f(z,X)$ for all nodes $z \not\in S \cup T$ (following the
notation~\eqref{eq:plus:minus}).  The quantities $f(X,z)$ and $f(z,X)$
are called the {\em incoming flow} and {\em outgoing flow} at the
node $z$.  The {\em value} of the flow is
$\texttt{value}(f)\defeq f(S,X) = f(X,T)$.  The problem asks for the
maximum value of a flow, which we denote by $\texttt{maxFlow}(G)$.  A
{\em cut} in the network is a set of edges\footnote{Usually the cut is
  defined as a set of nodes; in this paper we find it more convenient
  to define it as a set of edges.}  $C \subseteq X \times X$ whose
removal disconnects $S$ from $T$, and its {\em capacity} is the sum of
capacities of all its edges.  The max-flow, min-cut
theorem~\cite[Th.10.3]{schrijver-book} asserts that
$\texttt{maxFlow}(G)$ equals the minimum capacity of any cut.  Despite
significant algorithmic advances for the max-flow problem, see
e.g.~\cite{DBLP:conf/focs/Madry16}, practical algorithms are based on
the augmenting path method and remain slow in practice.  We show here
how to use the reduced graph of a quasi-stable coloring to compute an
approximate flow.  For that, we need to examine flows in bipartite
graphs.

When $G=(X,Y,c)$ is a bipartite graph, then we will assume that the
source nodes are $X$ and the target nodes are $Y$.  Obviously, the
maximum flow is the total capacity of all edges,
$\texttt{maxFlow}(G)=c(X,Y)$.  Next, we consider a restricted notion
of a flow.

\begin{definition}
  We say that a flow $f : X \times Y \rightarrow \R_+$ in a bipartite
  graph $G$ is {\em uniform} if all $\forall x_1, x_2 \in X$,
  $f(x_1,Y)=f(x_2,Y)$ and $\forall y_1,y_2 \in Y$,
  $f(X,y_1)=f(X,y_2)$; in other words, all source nodes have the same
  outgoing flow, and all target nodes have the same incoming flow.
\end{definition}

We denote by $\texttt{maxUFlow}(G)$ the max. value of a uniform
flow in $G$.

\begin{theorem} \label{th:flow:lower:upper:bound} Consider a network
  flow problem defined by $G =(X, c, \set{s}, \set{t})$, with a single
  source and a single target node, $s \neq t$.  Let
  $\bm P = \set{P_0 , P_1, \ldots, P_{k-1}, P_k}$ be any coloring,
  such that $P_0=\set{s}$ and $P_k=\set{t}$, i.e. the source and
  target nodes have their own unique colors.  Define two capacity
  functions on the reduced graph:
  \begin{align*}
    \hat c_1(i,j) \defeq & \texttt{maxUFlow}(P_i,P_j,c)
    & \hat c_2(i,j) \defeq & \texttt{maxFlow}(P_i,P_j,c)
  \end{align*}
  Let $\hat G_1, \hat G_2$ be the reduced graphs with nodes
  $\set{0,1, \ldots, k}$ and capacity functions $\hat c_1, \hat c_2$
  respectively.  Then:
  \begin{align*}
    & \texttt{maxFlow}(\hat G_1) \leq \texttt{maxFlow}(G) \leq \texttt{maxFlow}(\hat G_2)
  \end{align*}
\end{theorem}

\begin{proof}
  The second inequality follows immediately from the fact that the
  total amount of flow from a set $P_i$ to a set $P_j$ cannot exceed
  $c(P_i,P_j) = \hat c_2(i,j)$.  We prove the first inequality.  Fix
  any flow $\hat f$ in $\hat G_1$; we show how to construct a flow $f$
  in $G$ with the same value,
  $\texttt{value}(f)=\texttt{value}(\hat f)$.  The idea is to take the
  flow $\hat f(i,j)$ between any two colors $i,j$ of the reduced
  graph, and divided it uniformly between the nodes in $P_i$ and those
  in $P_j$.  For that purpose, we use the maximal uniform flow $f'$ in
  the bipartite graph $(P_i,P_j,c)$.  Since $\hat f$ satisfies the
  capacity condition, we have
  $\hat f(i,j) \leq \hat c_1(i,j)=f'(P_i,P_j)$.  Then, we define $f$
  on the bipartite graph $P_i,P_j$ to be equal to $f'$ scaled down by
  the factor $\hat f(i,j)/f'(P_i,P_j)$.  Then
  $f(P_i,P_j)=\hat f(i,j)$.  Importantly, $f$ is a uniform flow from
  $P_i$ to $P_j$, which means that all nodes $x \in P_i$ have exactly
  the same outgoing flow to $P_j$, namely $\hat f(i,j)/|P_i|$, and
  similarly all nodes $y \in P_j$ have the same incoming flow
  $\hat f(i,j)/|P_j|$.  This allows us to prove that $f$ satisfies the
  flow preservation condition on $G$ (since $\hat f$ is a flow on
  $\hat G$), and that $\texttt{value}(f)=\texttt{value}(\hat f)$.

\end{proof}

We use theorem to approximate the flow in a network as follows.
Compute a quasi-stable coloring, construct the reduced graph, set the
capacities $\hat c_2(i,j) = \sum_{x \in P_i, y\in P_j}c(x,y)$, and use
the upper bound in the theorem as an approximate value for the
max-flow.  The quality of this approximation depends on the how far
apart $\hat c_1$ and $\hat c_2$ are.  We show below in
Corollary~\ref{cor:stable:coloring:flow} that $\hat c_1=\hat c_2$ if
the reduced graph is defined by the stable coloring.  However, if we
relax the coloring to be quasi-stable, then the upper bound can be
arbitrarily bad, as we show next.

\begin{figure}
  \centering
      \includegraphics[width=0.8\linewidth]{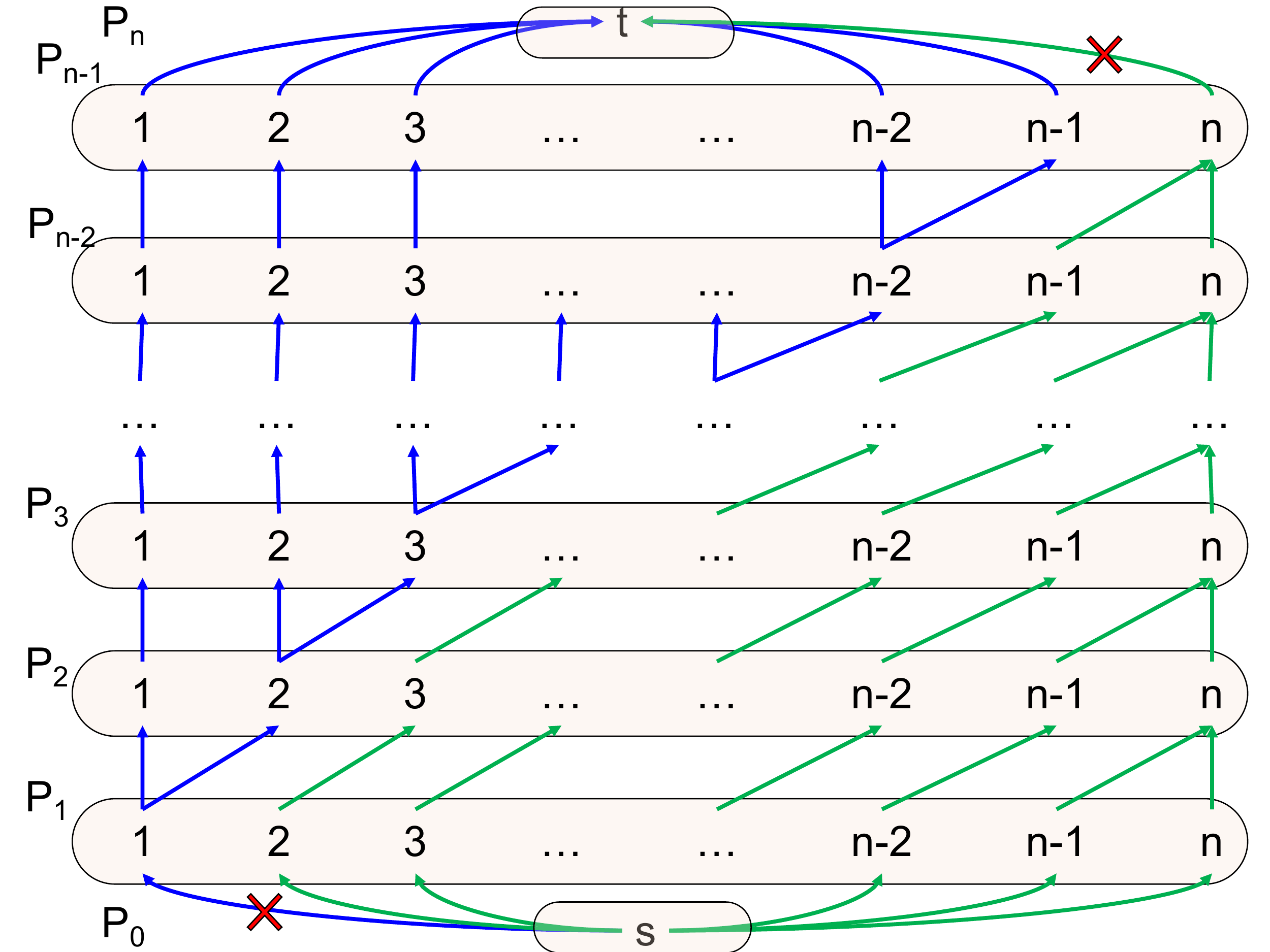}
      \caption{A network with $n^2+2$ nodes and a $q$-stable coloring
        with $q=1$.  The maximum flow is $2$, because there exists a
        cut of size 2 (the blue lower left and green upper right
        edges).  The maximum uniform flow of the bipartite graph
        induced by $P_{i-1}$ and $P_i$, for $i=2,n-1$, is 0, hence
        $\hat c_1=0$. The capacity between any two consecutive colors
        is $n$ or $n+1$ respectively, hence $\hat c_2(i-1,i) \geq n$.}
  \label{fig:uniform}
\end{figure}

\begin{example}\label{ex:no-flow}  Consider the network in
  Figure~\ref{fig:uniform}, where each edge has capacity 1.  The
  maximum flow is $2$, because there exists a cut with only two edges
  (lower left edge, and upper right edge in the figure).  The figure
  shows a coloring that is $q$-stable, for $q=1$; in other words this
  is a ``good'' quasi-stable coloring, as close as it can get to a
  stable coloring.  Let's examine the upper and lower bounds in
  Theorem~\ref{th:flow:lower:upper:bound}.  On one hand,
  $\hat c_2(i-1,i)=n+1$ for $i=2,n-1$, and
  $\hat c_2(0,1)=\hat c_2(n-1,n)=n$.  Therefore, the upper bound given
  by the theorem is $n$, which is a huge overestimate.  The reason is
  that the maximum {\em uniform} flow from $P_{i-1}$ to $P_i$ is 0.
  For example, if $f$ is a uniform flow from $P_1$ to $P_2$, then
  $f(1,1)+f(1,2)=f(2,3)$ (uniformity at nodes $1, 2 \in P_1$) and
  $f(1,1)=f(1,2)=f(2,3)$ (uniformity at nodes $1,2,3 \in P_2$), which
  implies $f=0$.
\end{example}

Despite this negative example, we show that, under some reasonable
assumptions, the two bounds in the theorem can be guaranteed to be
close:

\begin{lemma} \label{lemma:uniform} Let $G=(X,Y,c)$ be a bipartite
  graph, with capacity $c(x,y) \geq 0$.  Let $a, b > 0$ be two numbers
  such that, for all $x \in X$, $c(x,Y)\geq a$ and for all $y \in Y$,
  $c(X,y)\geq b$, and denote by
  $F \defeq \min(a \cdot |X|, b \cdot |Y|)$.  Assume that for any two
  sets of nodes $S \subseteq X$, $T \subseteq Y$, the following
  holds:\footnote{The condition is somewhat similar to Hal's marriage
    theorem~\cite[Th.16.7]{schrijver-book}.}
  \begin{align}
    c(S,T) + F \geq & a\cdot |S| + b \cdot |T| \label{eq:cond:a:b}
  \end{align}
  Then $\texttt{maxUFlow}(G)=F$.
\end{lemma}

We prove the lemma in Appendix~\ref{app:proof:lemma:uniform}.  Here,
we show an application.  A bipartite graph $G=(X,Y,c)$ is {\em
  $(a,b)$-biregular} if $c(x,Y)=a$ and $c(X,y)=b$ for all $x, y$.  We
show:

\begin{corollary} \label{cor:stable:coloring:flow}
  (1) If $G$ is an $(a,b)$-biregular graph, then
  condition~\eqref{eq:cond:a:b} holds.  (2) If $\bm P$ is stable
  coloring of a network $G$, then $\hat c_1=\hat c_2$ and the two
  bounds in Theorem~\ref{th:flow:lower:upper:bound} are equal.
\end{corollary}

\begin{proof}
  (1) In a biregular graph, the quantity $F$ defined in the lemma is
  $F = \min(a\cdot |X|, b \cdot |Y|) = a \cdot |X|=b \cdot
  |Y|=c(X,Y)$, and:
  \begin{align*}
    \begin{array}{ll}
    \multicolumn{2}{l}{F=c(S,T)+c(S,Y-T)+c(X-S,T)+c(X-S,Y-T)} \\
    a\cdot |S|=c(S,T)+c(S,Y-T)&b\cdot |T| = c(S,T)+c(X-S,T)
    \end{array}
  \end{align*}
  Condition~\eqref{eq:cond:a:b} simplifies to $c(X-S,Y-S)\geq 0$,
  which is true since all edge capacities are $\geq 0$.

  (2) If $\bm P$ is a stable coloring of a network, then every
  bipartite graph $(P_i, P_j, c)$ is $(a,b)$-biregular for some $a,b$
  and, furthermore
  $\hat c_1(i,j) = \texttt{maxUFlow}(P_i,P_j,c)=a\cdot
  |X|=c(P_i,P_j)=\hat c_2(i,j)$, proving the claim.
\end{proof}

\subsection{Centrality}

\begin{figure}
  \includegraphics[trim={0 4cm 0 0.85cm},clip, width=0.8\linewidth]{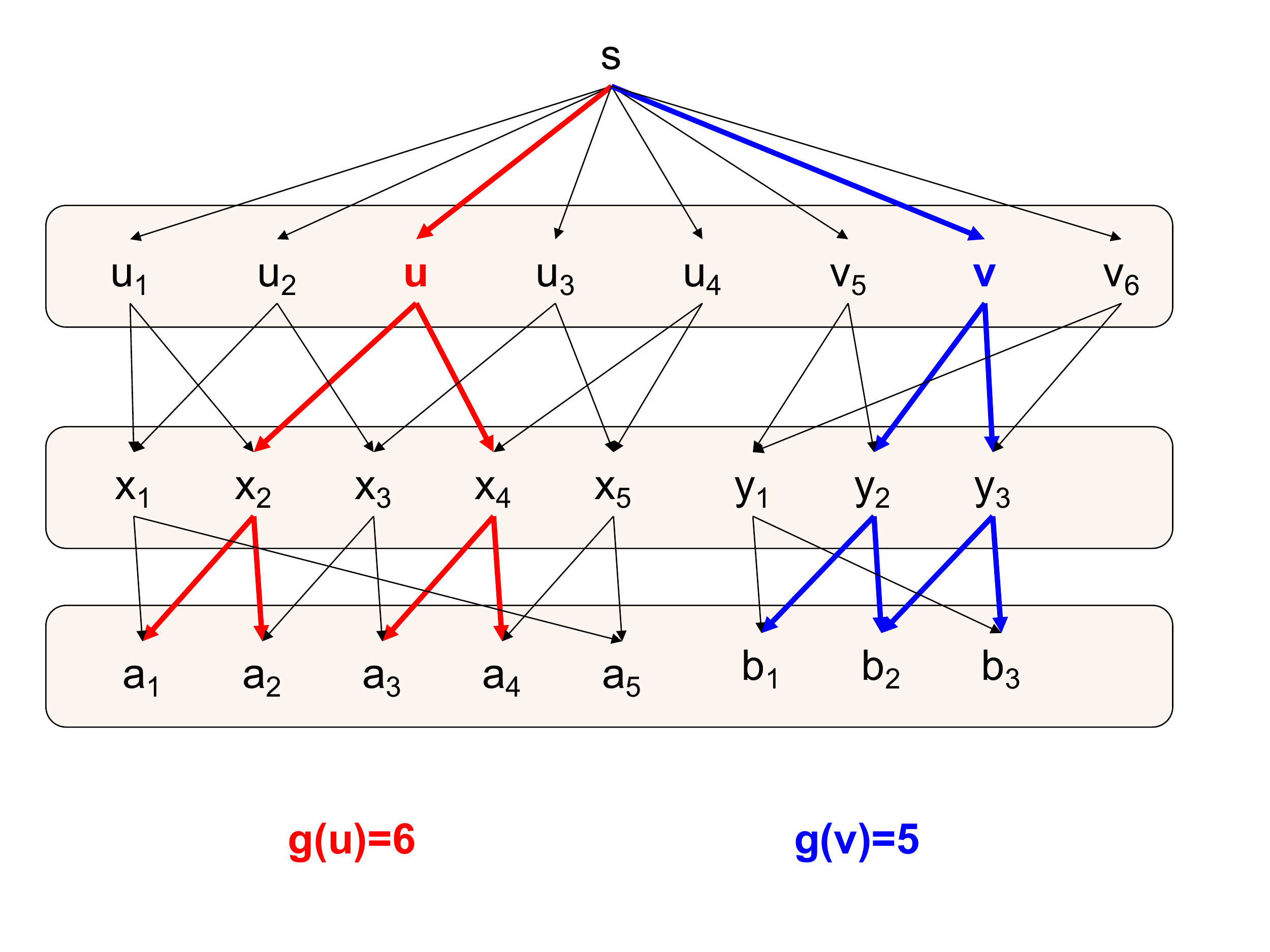}
  \caption{The two nodes $u$ and $v$ have the same color, but
    different betweenness centrality values: $\color{red} g(u)=6 \color{black}, \color{blue}g(v)=5$.}
  \label{fig:no-centrality}
\end{figure}

Finally, we consider the betweenness centrality in a graph and show two
results.  The first is negative, showing that, even if we compute a
stable coloring, nodes with the same color may have different
centrality values.  The second is positive, assuming we compute the
2-WL coloring instead of 1-WL.

The \textit{betweenness centrality} is a measure of influence for
graph vertices~\cite{10.2307/3033543}. The betweenness centrality of a
vertex $v$ is defined as:
\begin{align}
g(v) \defeq & \sum_{s, t: s \neq v \neq t \neq s}\frac{\sigma(s, t \mid v)}{\sigma(s, t)} \label{eq:centrality}
\end{align}
over all vertices $s, t$, where $\sigma(s, t)$ is the number of
shortest paths between $s, t$ and $\sigma(s, t \mid v)$ is the number of
those that pass through $v$.

We usually need to compute the centrality vector, consisting of the
values $g(v)$, for all nodes $v$.  To speed up this computation, we
first compute a quasi-stable coloring, then assume that all nodes of
in the same color have the same centrality value: this reduces the
cost of the computation, since we only need to
compute~\eqref{eq:centrality} once for each color (by randomly
sampling some $v$ in that color).  The question is how reasonable is
the assumption that nodes with the same color have similar centrality
values.

We observe that, even if the coloring is stable, two nodes $u,v$ of
the same color do not necessarily have the same centrality value.
This is shown in Fig.~\ref{fig:no-centrality}, where nodes $u$ and $v$
have the same color, but their centrality values differ.

However, we prove a positive result that still justifies our
heuristics.  Recall that the stable coloring consists of a partition
of the nodes of the graph, also called the 1 Weisfeiler-Lehman method,
or 1-WL.  We prove that, if two nodes are equivalent under the 2-WL
equivalence, then they have the same centrality value.  We refer the
reader to~\cite[pp.9]{DBLP:journals/corr/abs-2112-09992} for the
definition of 2-WL (and, more generally, of k-WL), but instead use the
following beautiful characterization of k-WL proved by Cai, F\"urer,
and Immerman~\cite[Th.5.2]{DBLP:journals/combinatorica/CaiFI92}, which
we review here in a slightly simplified form:

\begin{theorem} \label{th:cai:fuhrer:immerman}
  Let $C^{k+1}$ be the logic obtained by (a) extending First Order
  Logic with counting quantifiers of the form
  $\exists^{\geq m}x \varphi$, which means ``there are at last $m$
  distinct values $x$ that satisfy $\varphi$'', and (b) restricted to
  use only $k+1$ variables.  Then two nodes $a, b$ in a graph have the
  same $k$-WL color iff they satisfy the same $C^{k+1}$ formulas.
\end{theorem}

We prove in Appendix~\ref{app:proof:th:centrality}:

\begin{theorem} \label{th:centrality}
  Let $u,v$ be two nodes in a graph that have the same 2-WL color.
  Then they have the same centrality.
\end{theorem}

We anticipate further applications for our compression. Promising
problems to approximate are those whose solutions are robust to edge perturbations, capturing graph-wide
properties, including: clustering, node embedding, and computing graph layouts.

\section{Algorithm}

\label{sec:algo}

We have defined two variants of quasi-stable colorings, which allow us
to trade off the degree of stability (e.g. by varying $q$ or
$\varepsilon$) for the compression ratio (number of colors).  It turns
out that computing a quasi-stable coloring is more difficult than
computing the traditional stable coloring, for both variants.  We
will describe the challenge first, then introduce our proposed
algorithm.

\begin{figure}
  \centering
  \includegraphics[clip, trim=0cm 5.5cm 0cm 0cm, width=0.6\linewidth]{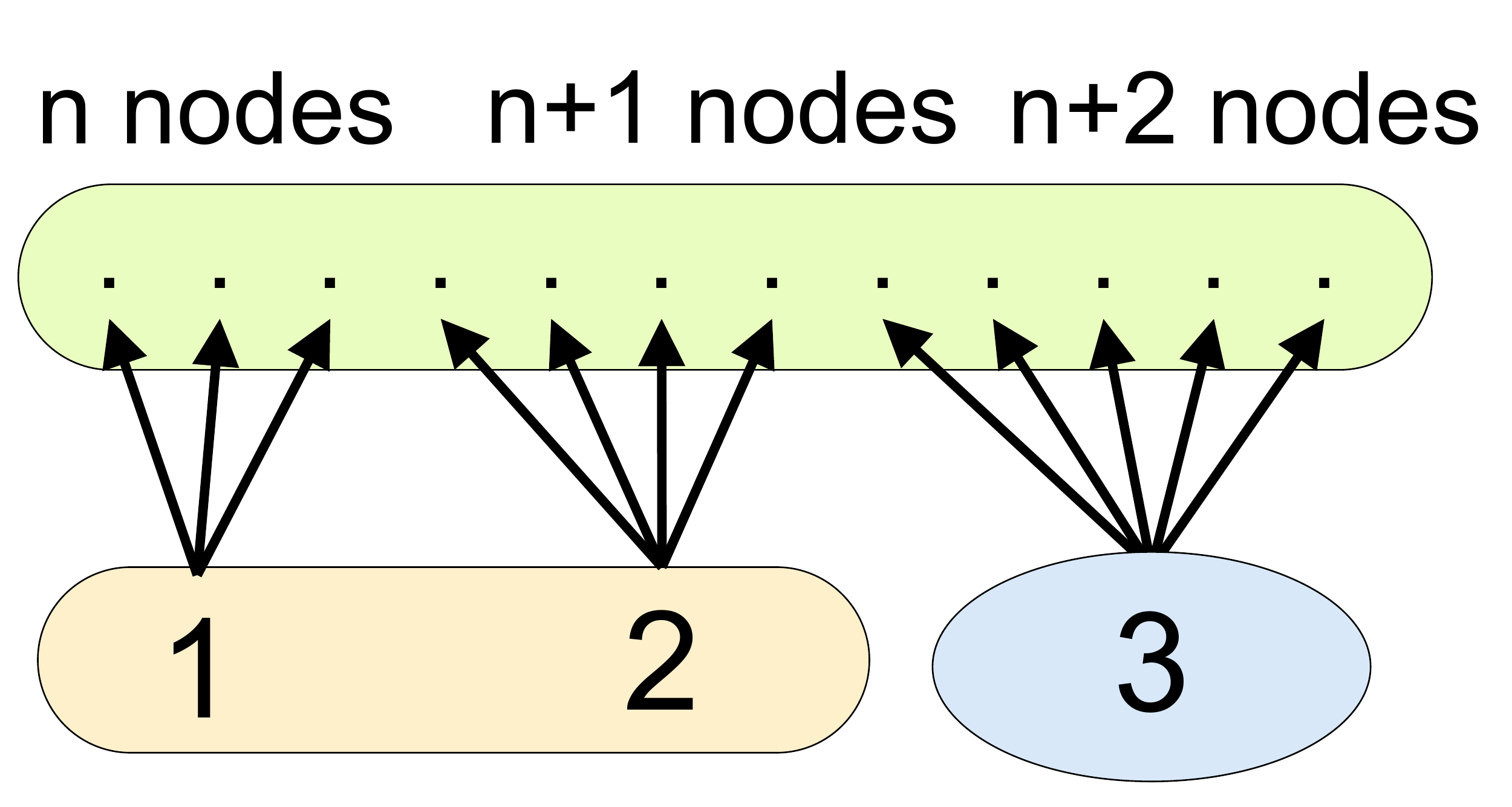}
  \caption{A graph with two maximal q-colorings, for $q=1$, or two
    maximal $\varepsilon$-relative coloring, for $\varepsilon=1/n$.
    All nodes at the top are in the same color (green) since all have
    exactly one incoming edge.  The nodes at the bottom can be
    partitioned into either $\set{1,2}$ and $\set{3}$ (as shown) or
    $\set{1}$ and $\set{2,3}$.  Both are maximal 1-stable colorings
    (because the degrees of two nodes differ by at most 1), and also
    $1/n$-relative colorings (because the relative error is at most
    $(n+1)/n \leq e^{1/n}$).}
  \label{fig:no-lub}
\end{figure}

\subsection{Complexity}
\label{sec:algo:cmplx}

The notion of {\em stable} coloring has the elegant property that
every graph has a unique, maximal stable coloring.  We show here that
this property fails for quasi-stable. We use standard
  terminology from partially ordered sets and call a valid coloring
  \textit{maximal} when no valid coarsening exists, i.e. it cannot be
  greedily improved, and call it \textit{maximum} or \textit{greatest
    element} when it is a coarsening of all valid colorings.
  Equivalently, a valid coloring is maximum if and only if it is the
  unique maximal valid coloring. Consider the graph in
Fig.~\ref{fig:no-lub}: there are two distinct maximal 1-stable
colorings, and also two distinct $1/n$-relative colorings, because we
can partition the nodes $1,2,3$ either as $\set{1,2}, \set{3}$ or as
$\set{1}, \set{2,3}$, but cannot leave them in the same color.  In
fact, we prove:

\begin{theorem} \label{th:lub} (1) If $\sim$ is a congruence
  w.r.t. addition (i.e. an equivalence relation satisfying
  $x \sim y \Rightarrow (x+z)\sim (y+z)$) then any graph admits a
  unique maximum $\sim$quasi stable coloring, which can be computed in
  PTIME.  (2) Computing a maximal q-stable coloring is NP-complete,
  and similarly for an $\varepsilon$-stable coloring.
\end{theorem}

For a simple illustration, fix $c \geq 0$ and define $x \sim y$ if
$\min(x,c) = \min(y,c)$; then $\sim$ is a congruence.  The theorem
implies that there is a unique maximal $\sim$quasi stable coloring.
When $c=1$ then this is the maximal bisimulation, and when
$c = \infty$ then it is the stable coloring.

The theorem implies that while finding a quasi-stable coloring
  is trivial (a unique color per node suffices), finding a
  coloring that cannot be improved is difficult. The question of the practicality of finding a
  ``good enough'' coloring is addressed in
  subsection~\ref{sec:rothko-algo}.

\begin{proof}
  (1) Assume $\sim$ is a congruence.  We first prove that, if
  $\bm P, \bm Q$ are $\sim$-stable colorings of a graph, then so is
  $\bm P \vee \bm Q$.  A color $C$ of $\bm P \vee \bm Q$ can be
  characterized in two ways: (a) for any two nodes $x,x' \in C$, there
  exists a sequence $x_0 := x, x_1, x_2, \ldots, x_n := x'$ such that
  every pair $(x_{i-1},x_i)$ is either in the same color of $\bm P$,
  or the same color of $\bm Q$, and (b) $C$ is both a disjoint union
  of $\bm P$-colors, and a disjoint union of $\bm Q$-colors, and is
  minimal such.  Let $C, D$ be two colors of $\bm P \vee \bm Q$, let
  $x, x' \in C$, let $w=w(x,D), w'=w(x',D)$ be their outgoing weights
  to $D$, and let $x_0 := x, x_1, x_2, \ldots, x_n := x'$ be the
  sequence given by (a).  Fix $i=1,n$, and assume w.l.o.g. that
  $x_{i-1},x_i$ have the same $\bm P$-color.  Then we use the fact
  that $D$ is a union of $\bm P$-colors,
  $D = P_{j_1} \cup \cdots \cup P_{j_k}$, and observe that
  $w(x_{i-1},D)=w(x_{i-1},P_{j_1})+ \cdots + w(x_{i-1},P_{j_k})$ and
  $w(x_i,D)=w(x_i,P_{j_1})+ \cdots + w(x_i,P_{j_k})$.  Since $\bm P$
  is $\sim$-stable, we have
  $w(x_{i-1},P_{j_\ell})\sim w(x_i,P_{j_\ell})$ for all $\ell$, which
  implies $w(x_{i-1},D)\sim w(x_i,D)$ because $\sim$ is a congruence.
  Finally, we derive $w(x,D)=w(x',D)$ because $\sim$ is transitive.
  Thus proves the claim that $\bm P\vee \bm Q$ is $\sim$stable.
  Finally, let $\bm P_1, \bm P_2, \ldots$ be all $\sim$stable
  colorings.  Then $\bm P_1 \vee \bm P_2 \vee \cdots$ is the unique
  maximum $\sim$stable coloring, and can be computed in PTIME using
  color refinement.

  (2) By reduction from the 2-dimensional Geometric Set Cover problem,
  more specifically from BOX-COVER, which is
  NP-complete~\cite{DBLP:journals/ipl/FowlerPT81}: we are given a set
  of points $S = \set{(a_1,b_1), \ldots, (a_n,b_n)} \subseteq \R^2$,
  with integer coordinates, and are asked to cover it with a minimum
  number of squares with a fixed width $q$.  Given this instance of
  BOX-COVER, we construct the following 3-partite graph $(X,Y,Z,E)$.
  All edges go from $X$ to $Y$ or from $Y$ to $Z$.  The set $Y$ has
  $n$ nodes.  Each node $y_i \in y$ has exactly $a_i$ incoming edges
  from $X$, and exactly $b_i$ outgoing edges to $Z$: thus
  $|X|=\sum_i a_i$ and $|Z|=\sum_i b_i$.  Any $q$-stable coloring of
  the graph corresponds to a cover of $S$ with $q \times q$ squares.
\end{proof}

\subsection{\textsc{Rothko} Algorithm}
\label{sec:rothko-algo}

Given the negative results above, we settle for a heuristic-based
algorithm for finding quasi-stable colorings that may not be
  maximal.  The main idea is that, instead of imposing some $q$, the
algorithm repeatedly applies a variant of color refinement, until a
maximum number of colors is reached.  The value of $q$ is the computed
on this coloring. The guarantees given by the theorems in
  Sec.~\ref{sec:approximate} still hold on the resulting colored
  graph, but the quality of the approximation depends on how good the
  value $q$ is when the algorithm terminates.
 We briefly describe the algorithm next.

We call the algorithm \textsc{Rothko} and list it in
Algorithm~\ref{alg:approx}. The name reflects its method of dividing
matrices into a few large, rectangular regions and giving them distinct
colors---resembling Rothko's famous color
field paintings.

This iterative algorithm
operates by refining one color at a time, beginning with the coarsest
(\textit{i.e.}, single color) partition. At every step a \textit{witness} is
identified: that is, the pair of partitions $P_i, P_j$ that maximize error
in the $P_i \rightarrow P_j$ direction. The source color is then split into
$P_i', P_i''$ to reduce the sum of errors $P_i' \rightarrow P_j$ and $P_i''
\rightarrow P_j$. This process is repeated until the desired error bound or number of
colors is reached.

A witness is identified by calculating the maximum, minimum degrees between
all colors ($U, L$ respectively)
and taking the difference of the two matrices. This produces the error
matrix $Err=U-L$, whose $(i, j)$\textsuperscript{th} entry records the q-error of color $P_i$
with respect to $P_j$. Then, a threshold is
set by taking the mean degree into $P_j$ of the nodes in $P_i$. The
elements of $P_i$ are then split depending on whether
their degree exceeds this threshold. We break ties
arbitrarily--notice if we have a tie, often at the next around we will
split the other tied-with color. We find ties to be rare in our experiments.

For some applications, it is desirable to weight the error by the size of
the partition, so that a q-error in a large partition is considered worse
than a q-error in a small partition. As such, the algorithm accepts two 
parameters, $\alpha, \beta$ which allow for building a weight matrix $C$.
These parameters control the weight assigned to the source and target
colors, respectively.
$C$ is multiplied element-wise by $E$ to produce the weighted error matrix $E_{weighted}$.
In practice, we set $\alpha = \beta = 0$ for max-flow problems, as neither the number of source
nor target nodes affects the flow, but only the total edge capacity between the
colors; for linear programs, $\alpha = 1, \beta = 0$ which
prioritizes colors with more rows; for betweenness centrality, $\alpha =
\beta = 1$ as the number of paths depends on both the number of nodes in
source and target color.

Further, for applications where all weights are non-negative, we observe that using the
geometric rather than arithmetic mean results in a more compact coloring.
The intuition can be gleaned from scale-free networks, where the
proportion of nodes with degree $k$ tends is proportional to $k^{-
\gamma}$. Under the most common scale-free model, 
Barabási–Albert~\cite{barabasi1999emergence}, where $\gamma = 3$ and the
average degree 
is $2m$, splitting using arithmetic mean yields unbalanced partitions with
$1 / (8m^3)$ fraction of nodes. Even for modest values of $m$ this
quickly becomes unbalanced (i.e. when $m = 3$ the partition will be split
$1 : 216$). Since the geometric mean is equivalent to the arithmetic
mean in log-space, the split is much less unbalanced (in the previous
example, it would be $1 : 4$). Many natural networks---such as the
internet---are
 thought to be scale-free~\cite{barabasi2013network}.

\begin{algorithm}
  \caption{\textsc{Rothko}\\Computing an approximate partition over a weighted graph $G$, with $n$ colors or
  $\varepsilon$ maximum $q$-error}\label{alg:approx}
  \KwData{$G = (V, E), W : V \times V \rightarrow \mathbb{R}^+$}
  \KwParameter{$n \in \mathbb{Z}^+, \varepsilon \in \mathbb{R}_{\geq 0},
  ~\alpha, \beta \in \mathbb{R}$}
  \KwResult{$P \subset \mathcal{P}(V) $}
  $P\gets \set{ V }$\;
  \While{$|P| < n$}{
    $U_{ij}, L_{ij} \gets \max_{v \in P_i} \text{deg}(v, P_j), \min_{v \in P_i} \text{deg}(v, P_j)$\;
    $Err \gets U - L$\;
    \If{$\max Err \leq \varepsilon$}{break\;}
  $C_{ij} \gets |P_i|^\alpha \times |P_j|^\beta$;  \tcp*[f]{weights}
\\  
    $Err_{\text{weighted}}\gets Err \odot C$; \tcp*[f]{element-wise
    product}\\
    $i, j \gets \argmax_{i, j} Err_{\text{weighted}}$;\tcp*[f]{witness}\\
    $\text{threshold} \gets \text{mean}(\set{x \mid x=\text{deg}(v, P_j), v \in P_i})$ \;
    \tcp{Split $P_i$ at threshold}
  $P_{\text{retain}} \gets \set{ v \in P_i \mid \text{deg}(v, P_j) \leq
  \text{threshold}}$\;
    $P_{\text{eject}} \gets P_i \setminus P_\text{retain} $\;
    $P \gets P \setminus \set{P_i} \cup \set{P_{\text{retain}}, P_{\text{eject}}}$\;
  }
\end{algorithm}

\textsc{Rothko} is an anytime algorithm. It can be interrupted and will still
produce in a valid coloring. The longer it is allowed to run, the better
the resulting coloring is. Further, the stopping condition can be set
depending on a desired number of colors, or target $q$-error, encoded in
Algorithm~\ref{alg:approx} as $n, \varepsilon$ respectively. This is
particularly valuable in interactive applications, where \textsc{Rothko} can be run
as a co-routine, with the application alternating between color refinement
and updating its approximation based on the new colors.

This algorithm is guaranteed to terminate. At each iteration, a color is
chosen to be split. Singleton colors are never selected for splitting, as
their degree-difference into any partition is zero. As such, after enough
iterations either the algorithm will reach its desired error bound, or have
refined into the trivial partitioning---consisting of only singleton
partitions and so a maximum $q$-error of zero.

\section{Evaluation}

\label{sec:evaluation}

We empirically evaluate our notion of quasi-stable coloring,
addressing three questions:

\begin{enumerate}
  \item End-to-end performance: how good is the system downstream? Can it
  effectively trade-off accuracy for speedup?
\item What are the characteristics of the colors? What are the
  properties of the compressed graphs?
\item How efficient and scalable is the \textsc{Rothko} algorithm?
\end{enumerate}

We use 20 datasets for evaluation. Graphs are
outlined in Table~\ref{tab:graphs}, linear programs in
Table~\ref{table:lp}. We list the primary sources in the table, many of these graphs
were found via dataset repositories~\cite{snapnets, waterloo-vision}. Trials are run on a MacOS machine with a 3.2~GHz ARMv8
processor and 16GB of main memory.  A single core is used for all
experiments. All code is run on \textit{Julia} \texttt{v1.7}, with
linear programs being solved with the \textit{Tulip} solver
and max-flow problems with the
\textit{GraphsFlows} library. Tulip is the fastest open-source solver~\cite{DBLP:journals/mpc/TanneauAL21}, while
GraphsFlows uses the state-of-the-art push-relabel algorithm~\cite{DBLP:conf/esa/Goldberg08}. Our coloring implementation, as tested in this paper,
is packaged as \texttt{QuasiStableColors} version \texttt{v0.1.0} and is available
for download.\footnote{\url{https://github.com/mkyl/QuasiStableColors.jl}}

\begin{figure*}
  \centering
  \begin{subfigure}[b]{0.3\textwidth}
  \includegraphics[width=\textwidth]{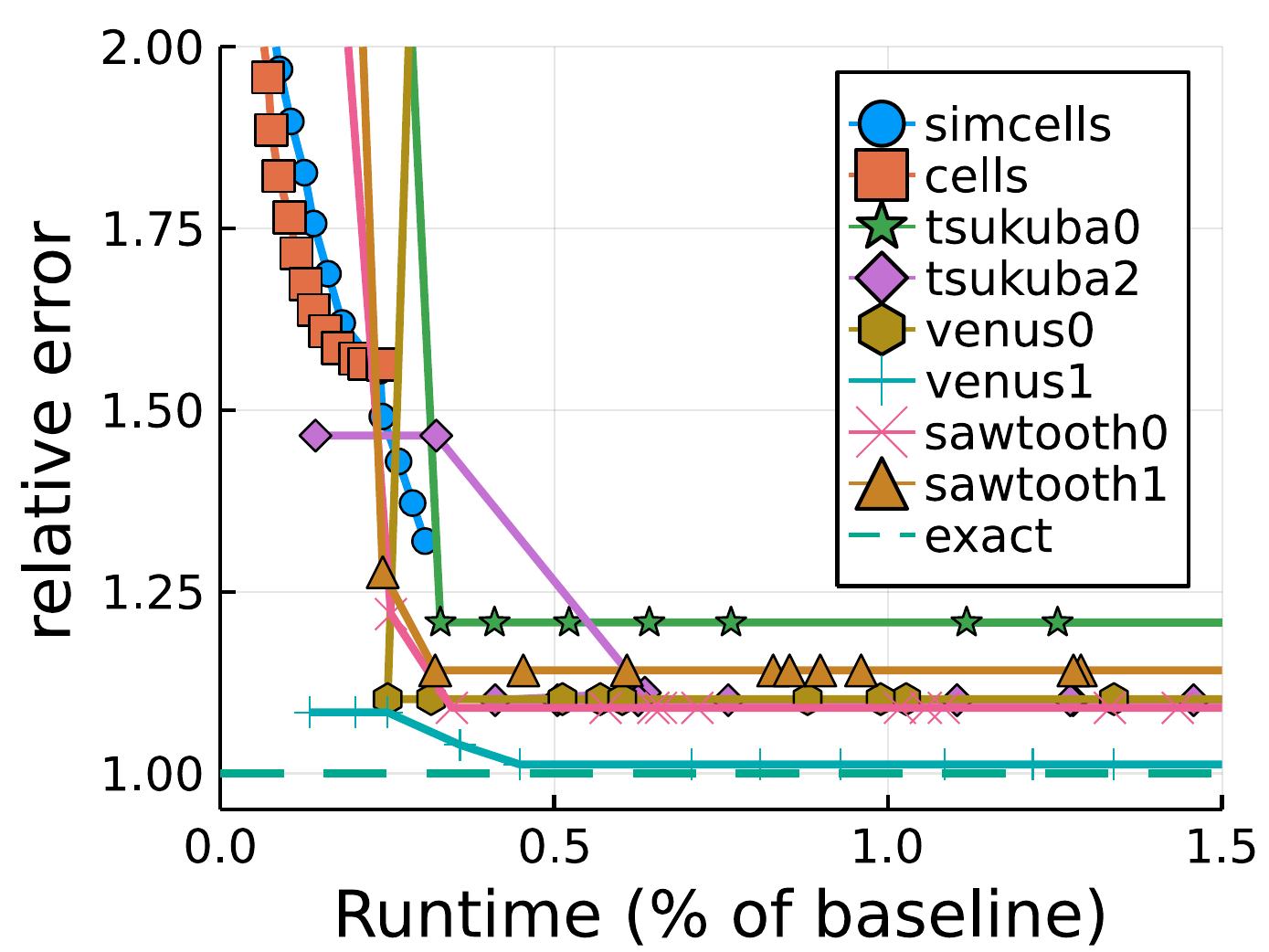}
  \caption{Maximum-flow}
  \end{subfigure}
  \hfill
  \begin{subfigure}[b]{0.3\textwidth}
    \includegraphics[width=\textwidth]{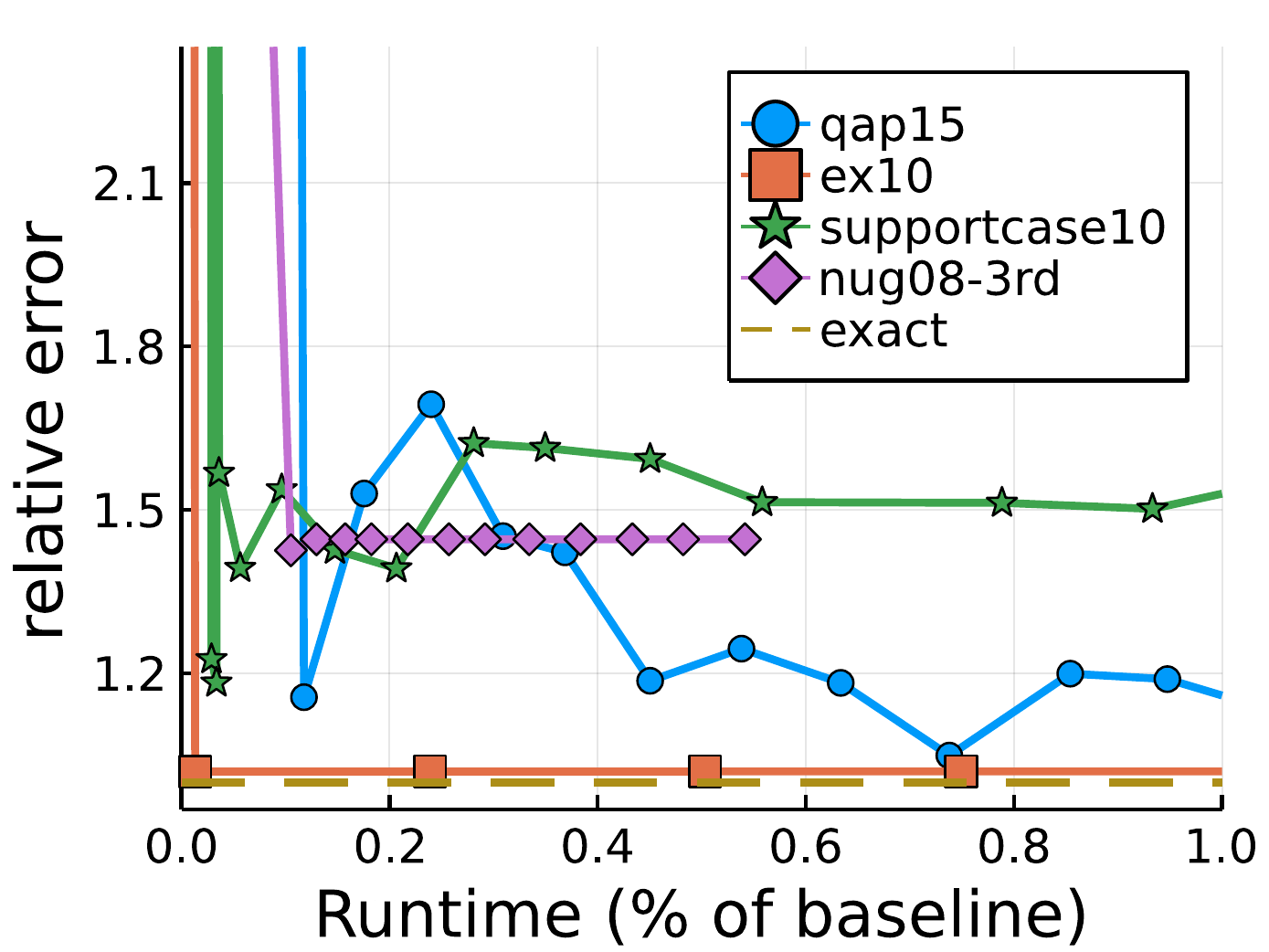}
    \caption{Linear optimization}
    \label{fig:lp}
  \end{subfigure}
  \hfill
  \begin{subfigure}[b]{0.3\textwidth}
    \centering
    \includegraphics[width=\textwidth]{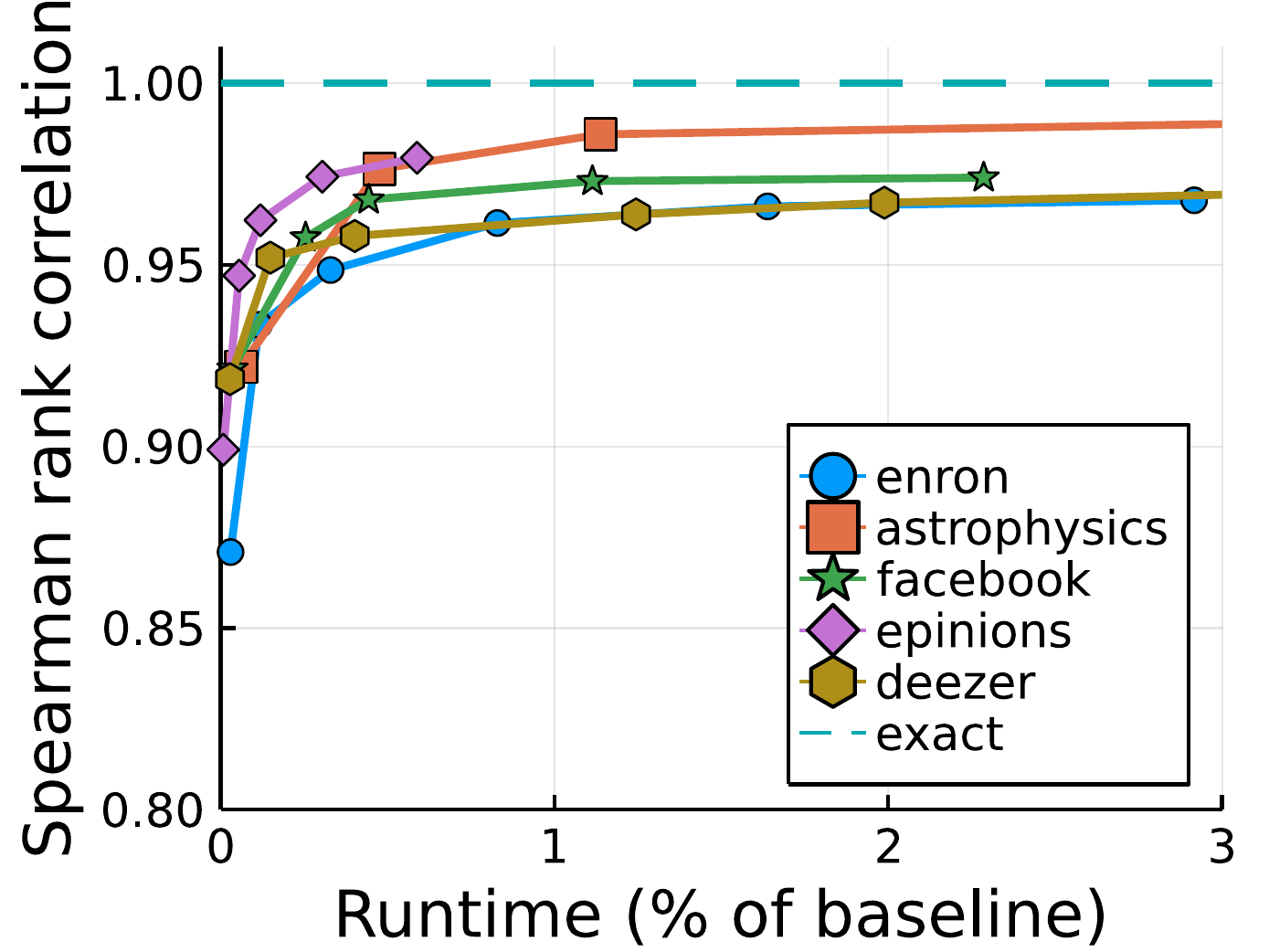}
    \caption{Centrality}
    \label{fig:centrality}
  \end{subfigure}
  \caption{Speed-accuracy trade-offs for three task types and 20 datasets. Runtime
  reported is end-to-end, including the time taken for graph coloring,
  building an approximate instance of the problem and solving it.}
  \label{fig:speed-accuracy}
\end{figure*}

\begin{figure*}
  \centering
  \begin{subfigure}[b]{0.3\textwidth}
  \includegraphics[width=\textwidth]{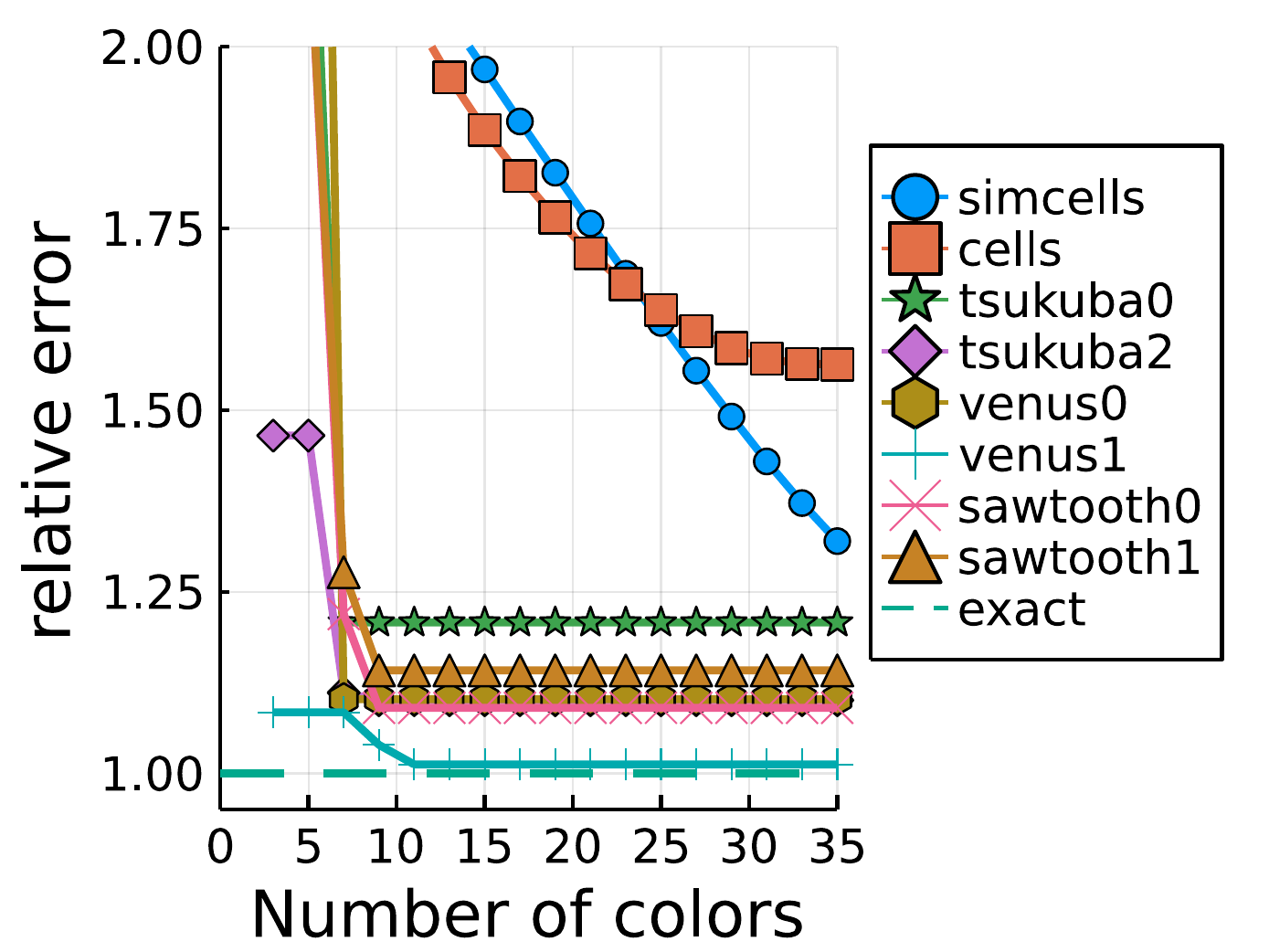}
  \caption{Maximum-flow}
  \end{subfigure}
  \hfill
  \begin{subfigure}[b]{0.3\textwidth}
    \centering
    \includegraphics[width=\textwidth]{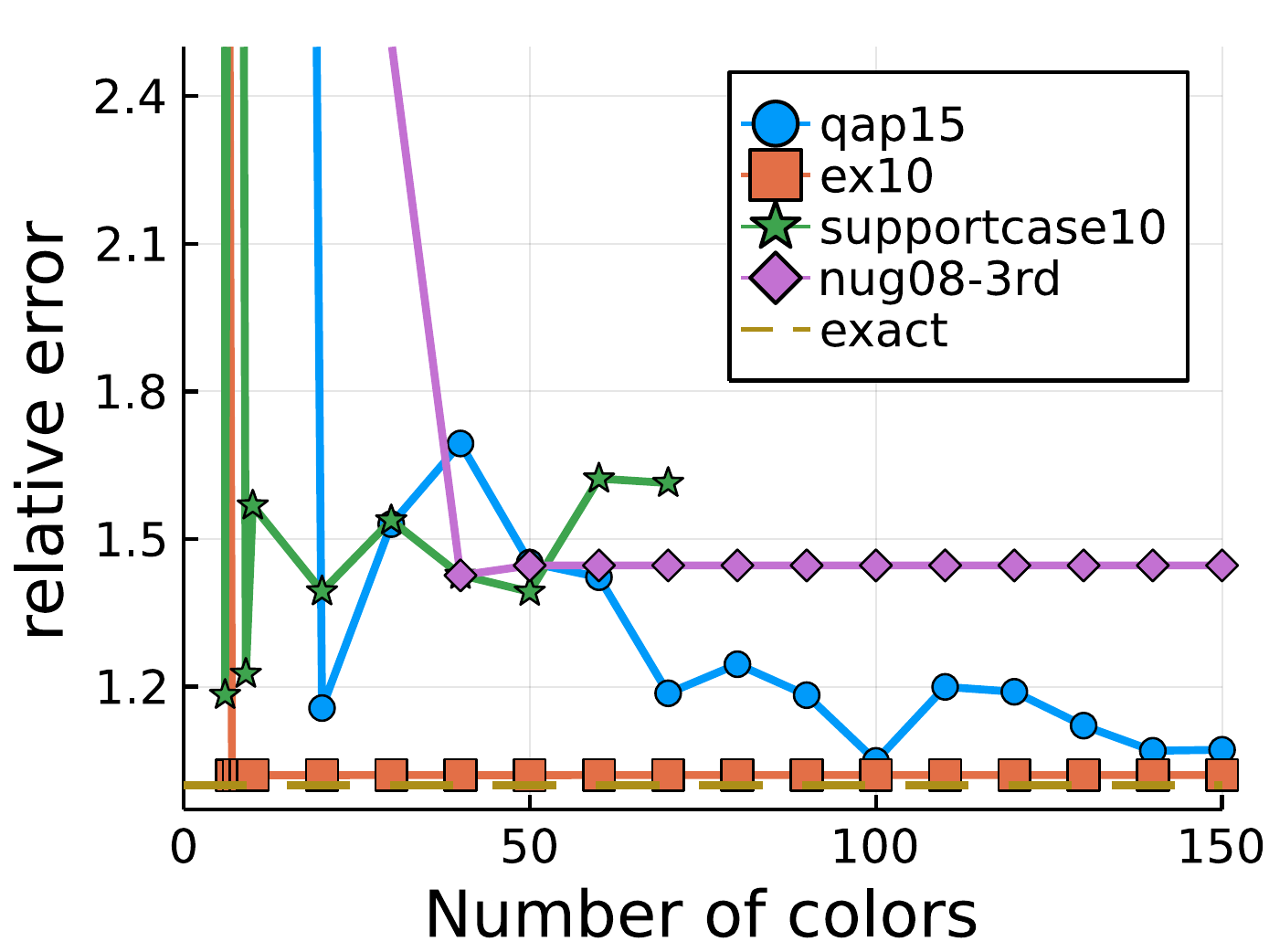}
    \caption{Linear optimization}
  \end{subfigure}
  \hfill
  \begin{subfigure}[b]{0.3\textwidth}
    \centering
    \includegraphics[width=\textwidth]{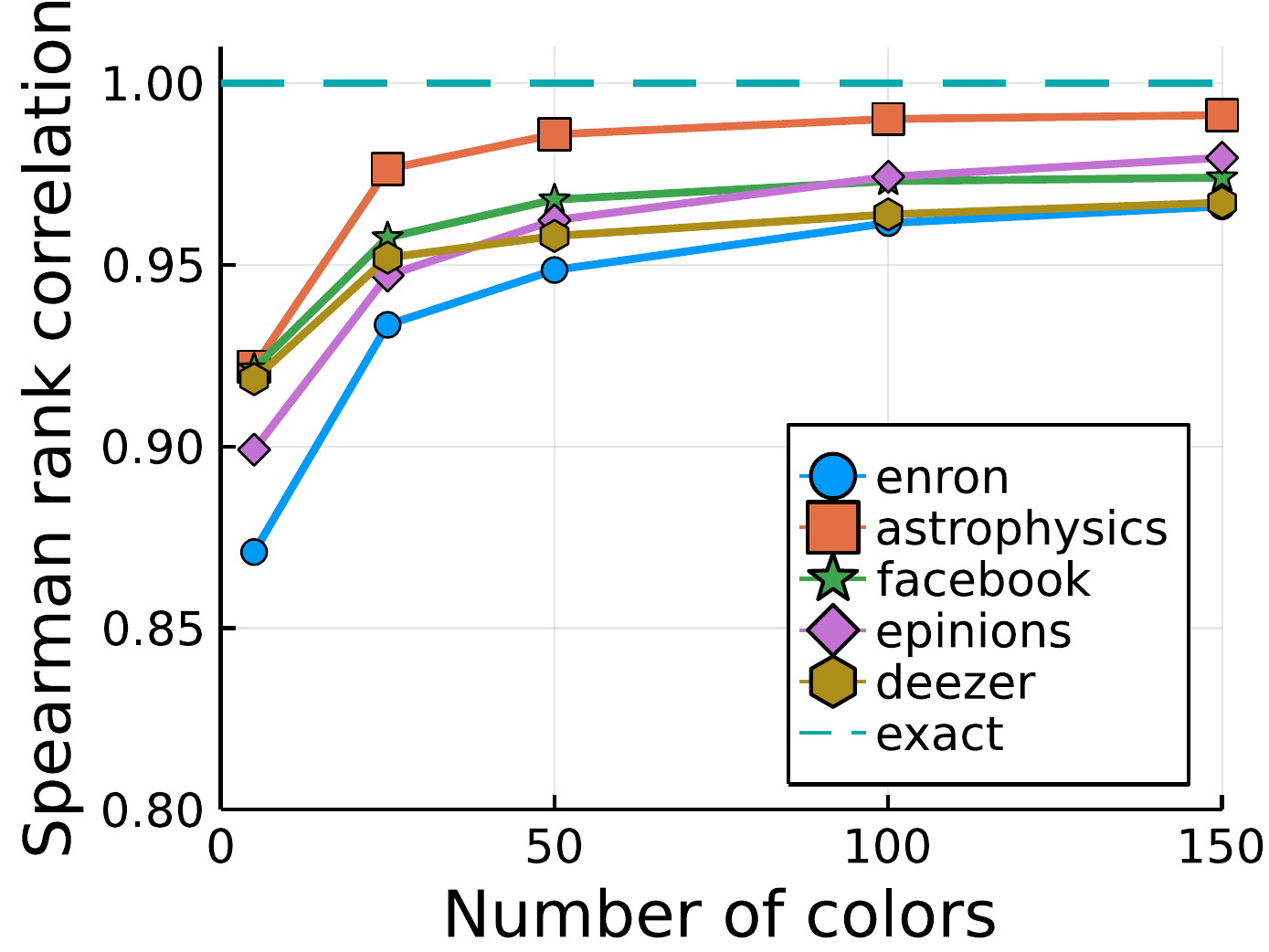}
    \caption{Centrality}
  \end{subfigure}
  \caption{Accuracy as a function of the number of colors, across the same three tasks.}
  \label{fig:color-accuracy}
\end{figure*}

\begin{table}
  \caption{Runtime comparison of q-stable colors
  vs. prior approximations
  (Riondato-Kornaropoulos~\cite{DBLP:conf/wsdm/RiondatoK14}
  and early-stopping~\cite{or-tools-time-limit}) and exact
  algorithms (Brandes~\cite{brandes2001faster} and interior-point solver~\cite{DBLP:journals/mpc/TanneauAL21}). Runtime to
 achieve a target approximation quality is measured; target is correlation ($\rho$) with ground truth
  values for centrality and relative error for linear optimization.
  ``$\times$'' is 20-minute timeout. Units in seconds, lower is better.}
  \label{tab:prior}
  \begin{tabular}{@{}lllllllc@{}}
  \midrule
  & \multicolumn{6}{l}{\textbf{Betweenness centrality}: ours,~\cite{DBLP:conf/wsdm/RiondatoK14}, and~\cite{brandes2001faster}} \\
  & \multicolumn{2}{l}{$\rho=0.90$} & \multicolumn{2}{l}{$\rho = 0.95$} &
   \multicolumn{2}{l}{$\rho=0.97$} & Exact \\
  & Ours& \multicolumn{1}{l|}{Prior}& Ours& \multicolumn{1}{l|}{Prior} & Ours& \multicolumn{1}{l|}{Prior} & \\ \midrule
  Astroph. &   \textbf{0.13}  &  \multicolumn{1}{l|}{15.2} &
  \textbf{1.03}          &  \multicolumn{1}{l|}{41.4}&\textbf{2.49} &
  \multicolumn{1}{l|}{61.1} & 223 \\
  Facebook     &  \textbf{0.07} & \multicolumn{1}{l|}{3.2}&\textbf{0.53}&
  \multicolumn{1}{l|}{7.1}& \textbf{2.23} & \multicolumn{1}{l|}{12.6} & 221\\
  Deezer       &   \textbf{0.05} & \multicolumn{1}{l|}{3.6}&\textbf{1.11}&
  \multicolumn{1}{l|}{7.2}&\textbf{8.56} & \multicolumn{1}{l|}{14.8} & 295 \\
  Enron        &    \textbf{0.41} & \multicolumn{1}{l|}{2.6} & \textbf{3.06} &
  \multicolumn{1}{l|}{5.6}& 10.8 & \multicolumn{1}{l|}{\textbf{8.7}} & 380 \\
  Epinions     &  \textbf{0.18} & 17.1 & \textbf{3.15}
  & 36.5 & \textbf{7.95} & 58.2 &  2552 \\
  \midrule
  & \multicolumn{6}{l}{\textbf{Linear optimization}: ours,~\cite{or-tools-time-limit}, and~\cite{DBLP:journals/mpc/TanneauAL21}} \\
  & \multicolumn{2}{l}{$\text{rel. err.}=3.0$} &
  \multicolumn{2}{l}{$\text{rel. err.} = 2.0$} &
   \multicolumn{2}{l}{$\text{rel err.} = 1.5$} & Exact\\
   & Ours& Prior& Ours& Prior& Ours& Prior\\ \midrule
  qap15 &   \textbf{3.20}  &  \multicolumn{1}{l|}{112.}&
  \textbf{4.91}          &  \multicolumn{1}{l|}{524.}&\textbf{11.4} &
  \multicolumn{1}{l|}{$\times$} & 1~320 \\
  nug08. &   \textbf{5.40}  &  \multicolumn{1}{l|}{1~027.}&
  \textbf{6.65}          &  \multicolumn{1}{l|}{$\times$} & \textbf{6.65} &
  \multicolumn{1}{l|}{$\times$} & 6~000\\
  support. &   \textbf{0.51}  &  \multicolumn{1}{l|}{143.}&
  \textbf{4.98}          &  \multicolumn{1}{l|}{143.}&$\times$ &
  \multicolumn{1}{l|}{$\times$} & 1~860  \\
  ex10 &   \textbf{0.247}  &  \multicolumn{1}{l|}{795.}&
  \textbf{14.0}          &  \multicolumn{1}{l|}{795.}&\textbf{14.0} &
  \multicolumn{1}{l|}{795.} & 1~440 \\ \bottomrule
    \end{tabular}
    \end{table}

\begin{table}
    \caption{Summary of graphs used for evaluation}
    \label{tab:graphs}
    \begin{tabular}{@{}lrrp{0.65cm}c@{}}
    \toprule
    \textbf{Name}        & Vertices & Edges & Real/ Sim. & Source \\ 
    {\em General evaluation} \\  \midrule
    \quad Karate & 34 & 75 & R & \cite{girvan2002community} \\
    \quad OpenFlights &  3~425   &  38~513  &  R & \cite{openflights} \\
    \quad \textsc{Dblp}        &  317~080 &  1~049~866  & R  & \cite{dblp}
    \vspace{0.1cm} \\
    {\em Centrality} \\  \midrule
    \quad Astrophysics & 18~772 & 198~110 & R & \cite{DBLP:journals/tkdd/LeskovecKF07}\\
    \quad Facebook & 22~470 & 171~002 & R & \cite{DBLP:conf/nips/McAuleyL12} \\
    \quad Deezer & 28~281 & 92~752 & R & \cite{feather} \\
    \quad Enron & 36~692 & 183~831 & R & \cite{DBLP:conf/ceas/KlimtY04} \\
    \quad Epinions & 75~879 & 508~837 & R &
    \cite{DBLP:conf/semweb/RichardsonAD03}\vspace{0.1cm} \\
    {\em Maximum-flow} \\  \midrule
    \quad Tsukuba0 &110~594 & 506~546 & R & \cite{tsukuba} \\
    \quad Tsukuba2 &110~594 & 500~544 & R & \cite{tsukuba}\\
    \quad Venus0 & 166~224 & 787~946 & R & \cite{DBLP:journals/ijcv/ScharsteinS02} \\
    \quad Venus1 & 166~224 & 787~716 & R & \cite{DBLP:journals/ijcv/ScharsteinS02} \\
    \quad Sawtooth0 & 164~922 & 790~296 & R & \cite{DBLP:journals/ijcv/ScharsteinS02}\\
    \quad Sawtooth1 & 164~922 & 789~014 & R &
    \cite{DBLP:journals/ijcv/ScharsteinS02}\\
    \quad SimCells & 903~962 & 6~738~294 & S & \cite{DBLP:conf/cvpr/JensenDD20} \\
    \quad Cells & 3~582~102 & 31~537~228 & R &
    \cite{DBLP:conf/cvpr/JensenDD20} \\ \bottomrule
    \end{tabular}
\end{table}

\begin{table}
  \caption{Summary of the linear programs used for evaluation. All instances
  are from real problems.}
  \label{table:lp}
  \begin{tabular}{@{}lrrp{1cm}p{0.65cm}c@{}}
  \toprule
  Name        & Rows & Cols. & Non-zeros & Sol. time & Source \\ \midrule
  qap15 & 6~331 & 22~275 & 110~700 & 22~min & \cite{Mittelmann} \\
  nug08-3rd &  19~728   &  20~448  & 139~008 & 100~min & \cite{Mittelmann} \\
  supportcase10 &  10~713 & 1~429~098  & 4~287~094 & 31~min & \cite{Mittelmann}\\
  ex10 & 69~609 & 17~680 & 1~179~680 & 24~min & \cite{Mittelmann}\\ \bottomrule
  \end{tabular}
\end{table}

\begin{table}
  \caption{Runtime and compression ratios of quasi-stable coloring \textit{vs.}
  prior work (stable coloring~\cite{L_Staudt_NetworKit_A_Tool_2016,DBLP:journals/mst/BerkholzBG17}) for selected datasets.}
  \label{tab:runtimes}
    \begin{tabular}{@{}lccrrr@{}}
    \toprule
      Dataset & Max $q$ & Mean $q$ & Colors & Compression & Time \\ \midrule
      OpenFlights        &   \multicolumn{2}{c}{stable ($q=0$)}  &  2~637  & 1.29:1 & 150ms      \\
      & q = 64  & 15.8 & 9 & 380:1 &  10ms \\
      & q = 32  & 6.96 & 17 & 200:1 & 20ms       \\
      & q = 16  & 2.22 & 39 & 87:1 & 60ms       \\
      & q = 8  & 0.52 & 106 & 32:1 & 350ms     \\
      \midrule
      Epinions        &   \multicolumn{2}{c}{stable ($q=0$)}  &  53~068 & 1.42:1 &
      49s \\
      & q = 64  & 4.42 & 71  & 1~000:1 &  2.39s      \\
      & q = 32  & 1.17 & 144 & 526:1 &  8.95s      \\
      & q = 16  & 0.79 & 316  & 240:1 & 40.5s \\
      & q = 8  & 0.22 & 869 & 87:1 & 5m19s \\ \midrule
    DBLP    & \multicolumn{2}{c}{stable ($q=0$)} &  233~466 & 1.35:1  &   14m52s      \\
    & q = 64  & 11.94 & 21 & 15~000:1 &  2.28s      \\
    & q = 32  & 2.22 & 89 & 3~500:1 & 22.6s    \\
    & q = 16  & 0.39 & 373 & 850:1 & 6m39s     \\
    & q = 8  & 0.06 & 1513 & 210:1 & 2h38m     \\ \bottomrule
    \end{tabular}
  \end{table}

\subsection{End-to-end performance} 
\label{sec:eval:end2end}

In this section, we consider how well q-stable colors work as a function of
downstream performance. We evaluate the trade-off between accuracy and
speed when using the colorings for approximating linear-optimization,
maximum-flow, and centrality tasks. On all tasks, we compare against the
baseline of solving the problem directly on the graph or linear system.

For maximum-flow and linear-optimization tasks, we use the  
relative error as the performance metric. We define this as
$\max(v / \hat{v}, \hat{v} / v)$ for an actual, predicted $v, \hat{v}$ so that
$1.0$ is the ideal score. For betweenness centrality, the actual and
predicted scores
are compared pair-wise using Spearman's rank correlation coefficient, where
$1.0$ is also the ideal score. In all
experiments, the time reported is end-to-end
\textit{i.e.}, includes coloring the graph or matrix, computing the
reduced problem, and solving it. 

Figure~\ref{fig:speed-accuracy} illustrates
the trade-off for three task types across twenty datasets. Overall, accuracy can be exchanged
for speed favorably: in the average case, using a budget of $1\%$ of the
baseline (\textit{i.e.,} a $100\times$ speedup) results in an average error
within $12\%$ of the optimal value.
Figure~\ref{fig:color-accuracy} shows
the number of colors required to achieve the same accuracy. Across all
tasks, no more than 150 colors are
required to converge to an approximation. We observe a consistent diminishing-returns pattern: the
initial color refinement results in large gains in accuracy, but as more
and more colors are added the gains shrink in size. Comparing
the densities of the graphs (omitted for brevity) with the difficulty of coloring them, we find no
consistent trend. Next, we consider the tasks individually.

\paragraph{Maximum Flow}
We test our algorithm on problem instances defined by min-cut/max-flow
benchmarks~\cite{patrick_m_jensen_2021_4905882, waterloo-vision}.
 This involves computing maximum flows
over the eight flow networks. We compare against a baseline of computing the exact flow
using the \textit{push-relabel algorithm}, considered to be the
benchmark for max-flow~\cite{DBLP:conf/esa/Goldberg08}.

Figure~\ref{fig:speed-accuracy}(a) shows the results:
our approximation achieves an average geometric-mean error of $1.17$ 
while using less than $1\%$ of the time needed for
direct solution. At the same time, as outlined in
Figure~\ref{fig:color-accuracy}(a), this error is achieved using no more
than $35$ colors. Recall that the flow networks are composed of 100K-2M nodes.

We consider comparing with prior approximation algorithms. 
The state-of-the-part push-relabel algorithm for max-flow cannot be stopped
early, as it computes \textit{pre-flows} which are not valid flows and
violate the principle of flow conservation. Further, while linear-time
approximations have been developed in the theory~\cite{DBLP:conf/soda/KelnerLOS14}, these algorithms
remain slower than push-relabel algorithm in practice.
\paragraph{Linear programs}

Next, we evaluate the ability of quasi-stable coloring to approximate solutions
to linear systems of equations. We test on four real-world linear
programs, outlined in Table~\ref{table:lp}. These are relatively difficult
tasks: the easiest can be solved in 20
minutes while the most difficult requires about two hours for an exact
solution. 

$q$-stable colors provide a good speed-accuracy tradeoff, shown in 
Figure~\ref{fig:speed-accuracy}(b). On average,
a geometric-mean relative-error of $1.13$ is reached in under $0.5\%$ of the
direct runtime. Figure~\ref{fig:color-accuracy}(b) shows the number of colors required for
the results. Similarly to max-flow, a relatively small number of colors is
required for an accurate answer. Unlike other tasks, the error on LPs is not monotone.

Table~\ref{tab:prior} (bottom) compares our approximation
with early stopping the interior-point-method solver, the recommended
approach in practice~\cite{or-tools-time-limit}. We set a relative error
and solve until that bound is met. Q-stable coloring
outperforms the baseline runtime by $10^2\times$ on average and times out in only one
configuration (\textit{vs.} five).

\paragraph{Centrality}

Next, we test the utility of q-stable colorings in approximating
betweenness centrality. We measure the
approximation error using Spearman's rank correlation coefficient. We
 compare against the baseline of solving for exact centralities
using Brandes algorithm~\cite{brandes2001faster}, the algorithm with the
best asymptotic runtime.

Figure~\ref{fig:centrality} shows the
speedup-accuracy tradeoff on five medium-size datasets. On all datasets,
the approximation is favorable: using $1\%$ of the time of the direct
computation, it produces centralities with correlation $0.973$ to the
ground truth. Figure~\ref{fig:color-accuracy} outlines the number of colors required for
these results. We find that using 50 colors is sufficient to ensure a rank
correlation of greater than $0.948$, while 100 colors allow for $0.965$. 
Recall the graphs have 18--75K vertices. We exclude the datasets
\textsc{dblp} and larger because the baseline timed out after 16 hours.

We note a few trends. First, the speed-accuracy trade-off is
more favorable the larger the dataset. The largest dataset,
\texttt{epinions}, shows the steepest slope at the beginning; the dataset
with the fewest vertices, \texttt{Astrophysics}, shows the shallowest
slope. As with maximum-flow, the approximation error is found to be monotone over all
datasets: the more colors used, the better the correlation is.

Table~\ref{tab:prior} (top) compares the performance of our approximation
with prior works~\cite{DBLP:conf/wsdm/RiondatoK14}. By compressing all
nodes in the graph
rather than focusing on selecting paths to sample, we obtain
a $30\times$ better average runtime across various tasks and approximation budgets.

\subsection{Coloring Characteristics}
\label{sec:eval:color-char}

\paragraph{Coloring size} Table~\ref{tab:runtimes} compares the
q-stable coloring of three datasets with stable coloring.
Stable coloring results in compressed graphs with $70\%$-$78\%$
of the full graph size. We find that small maximum values of $q$, such
as $q=8$ result in an order-of-magnitude improvement in the compression
ratios over stable coloring. Moderate maximum values of $q$, such as $q=16$ result
in a two or more orders-of-magnitude improvement. While the choice of $q$
caps the worst-case degree error, the average errors are much smaller,
on average $<1.0$ on all datasets: less than one differing edge per color.

\paragraph{Color distribution} Unlike stable coloring, single-element partitions do not
dominate any dataset. For example, on  \texttt{cells}, \texttt{DBLP}, 
\texttt{nug08-3rd} and \texttt{epinions}, the median partition
contains 6, 14, 56, 206 nodes in
each dataset respectively. This evidences
 the ability of $q$-stable colors to
avoid stable-coloring-like single-color partitions.

\begin{table}
  \caption{Characteristics of the constraint matrix for some compressed linear programs.}
  \label{table:compressed-lp}
  \begin{tabular}{@{}lp{0.75cm}rrp{0.75cm}p{1cm}p{0.75cm}@{}}
  \toprule
  Dataset     & Colors  & Rows & Cols. & Non-zeros & Comp-ression ratio &
  Rel. error \\ \midrule
  qap15 & 10 & 4 & 7  & 11 & $10^4$  & 18.91 \\
   & 50 & 27 & 24  & 179 & $10^3$  & 1.45 \\
   & 100 & 52 & 49  & 478 & $10^2$  & 1.05 \\ \midrule
  nug08-3rd & 5 & 3 & 3 & 5 & $10^4$ & 8.19  \\
   & 50 & 30 & 21 & 254 & $10^3$  & 1.45 \\
   & 100 & 61 & 40 & 930 & $10^3$ & 1.45 \\ \midrule
  supportcase10 & 5 &  3 & 3 & 4 & $10^5$ & $10^{10}$  \\
  & 50 &  30 & 21 & 131 & $10^3$ & 1.39 \\
  & 100 &  62 & 39  & 368 & $10^3$ & 1.51 \\\midrule
  ex10 & 5 & 5 & 1  & 4 & $10^6$ & 5.28 \\
  & 50 &  25 & 26 & 347 & $10^3$ & 1.02 \\
  & 100 & 51 & 50  & 864  & $10^3$ & 1.02 \\ \bottomrule
  \end{tabular}
\end{table}

\paragraph{Compression ratios}

Table~\ref{table:compressed-lp} shows the compression ratios enabled by
using quasi-stable colors on linear programs. A maximum compression of $10^6\times$
is recorded, but corresponds to a large error of 5.28. Typical space savings of
a ratio of
$10^2$-$10^3$ while maintaining a geometric mean error of $1.23$.
The outlier error of $10^{10}$ on \texttt{supportcase10}
is explained by the measured q-error of $10^{7}$ when only 5 colors are
used--this sharply decreases with more colors.

\subsection{Algorithm Properties}
\label{sec:evaluation:algorithm}
\begin{table}
  \caption{Average latency and responsiveness of the \textsc{Rothko} algorithm across task types.}
  \label{table:response}
  \begin{tabular}{@{}lp{1.5cm}p{1.35cm}p{1.5cm}@{}}
  \toprule
  Task        &  Time-to-first-result & Update frequency & Time to converge \\ \midrule
  Linear opt. & 560~ms & 2.71~s & 72.2~s \\
  Max-flow & 845~ms & 1.57~s  & 21.0~s  \\
  Centrality & 32~ms  & 1.60~s & 5.88~s \\ \bottomrule
  \end{tabular}
\end{table}

\paragraph{Runtime} Table~\ref{tab:runtimes} compares runtime against the state-of-the-art
stable-coloring algorithm~\cite{DBLP:journals/mst/BerkholzBG17} with complexity
$\mathcal{O}((n+m)\log n)$. \textsc{Rothko}'s runtime is competitive against this highly
optimized implementation~\cite{L_Staudt_NetworKit_A_Tool_2016}, with an order-of-magnitude better compression.

\paragraph{Responsiveness}
Because of the progressive nature of the \textsc{Rothko} algorithm, an
initial prediction is produced promptly. Table~\ref{table:response}
outlines the latency and responsiveness metrics. The first prediction
occurs within 480~ms on average, with centrality tasks having the consistently
lowest latency and max-flow the highest. Average latency is
strongly influenced by outliers, such as \texttt{cells} with 6.5s of
latency. Further, the algorithm
iterates well, with a new color computed every $1.96s$ on average. 
The time taken for the subtasks varies: for max-flow and linear-program
problems, the coloring step dominates, using $>99.9\%$ of the
runtime on  the measured datasets. For centrality the
solving step dominates, taking up $68\%$--$94\%$
of the runtime.

\paragraph{Robustness}
We compare the robustness of q-colors to graph perturbations against that
of stable coloring. We construct a synthetic graph $|V|=1000,E=|21~600|$ with a
compact, 100-color stable
coloring. 
Then in Figure~\ref{fig:graphs-stable-coloring}, a small number of edges is added at
random. The initial stable coloring has a compression ratio of $10\times$.
Perturbing $1.5\%$ of edges causes the stable coloring to degrade to a
compression ratio of 75\% (750 nodes, down from 1000), with a majority of nodes given a unique color.
Computing a $q$-stable color ($q=4$), the compression
ratio can be maintained at $6.5\times$ with the same perturbation.

\section{Conclusion}

\label{sec:conclusion}

We have introduced quasi-stable colorings, an approach for vertex classification
that allows for the lossy compression of graphs. By developing an
approximate version of stable coloring, we are able to practically color
real-world graphs. We show the ability of these colorings as
approximations of max-flow/min-cut, linear optimization and betweenness
centrality and prove their error bounds. Discovering that the construction
of maximal q-stable colorings is
NP-hard, we develop a heuristic-based algorithm to efficiently compute them. We
empirically evaluate the characteristics and approximation utility of
quasi-stable colors; validating their practicality on wide range of
real datasets and tasks.

\begin{acks}
    This project was partially supported by NSF IIS 1907997 and NSF-BSF 2109922.
\end{acks}
\appendix

\section{Appendix}

\label{sec:appendix}

\paragraph*{Proof of Theorem~\ref{th:approx:lp}}

\label{app:proof:th:approx:lp}

We prove here Theorem~\ref{th:approx:lp}.  In general, it is well
known that $\opt(A,b,c)$ is a continuous function in $A, b, c$.  We
prove here a stronger statement: if $A,b,c$ is well behaved (see
Sec.~\ref{sec:linear:optimization}), then the function is Lipschitz
continuous in $b, c$.

\begin{lemma} \label{lemma:lp:lipschitz} Given $A, b, c$ as above,
  there exists $q_0 > 0$ such that, forall $u \in \R^m, v \in \R^n$,
  if $||u||_\infty \leq q_0$ and $||v||_\infty \leq q_0$, then
  $|\opt(A, b+u, c+v) - \opt(A,b,c)| = O(||u||_\infty+||v||_\infty)$.
\end{lemma}

\begin{proof}
  Recall that the optimal solution $x^*$ can always be chosen to be a
  vertex of the polytope defined by the LP.  More precisely, consider
  the $m+n$ inequality constraints $Ax \leq b$, $x \geq 0$.  Choose
  any $n$ of them and convert them to equalities; if they uniquely
  define $x$, then we call $x$ a {\em candidate solution}, and we
  denote by $x_1, \ldots, x_N$ all candidate solutions.  Let $x_i$ be
  candidate solution that is feasible and optimal for $A,b,c$, and let
  $x_j$ be the candidate solution that is feasible and optimal for
  $A,b,c+v$ respectively.  Then
  $|\opt(A,b,c+v)-\opt(A,b,c)|=|v^T(x_j - x_i)|\leq ||v||_\infty
  ||x_j-x_i||_1 \leq O(q)$, where the constant in $O(-)$ is
  $2\max_i||x_i||_1$.  By applying the same argument to the dual LP we
  obtain $|\opt(A,b+u,c) - \opt(A,b,c)| = O(q)$.  Returning to the
  primal LP, we notice that if we replace $b$ with $b+u$, then the
  candidate solutions $x_i$ will change to some $x_i'$; since $u$
  ranges over a compact set, $\sup_u ||x_i'||_1$ exists and is finite,
  for all $i=1,N$, which implies
  $|\opt(A,b+u,c+v)-\opt(A,b+u,c)| = O(q)$.  Thus,
  $|\opt(A,b+u,c+v)-\opt(A,b,c)| = O(q)$ as required.
\end{proof}

Using the lemma, we can now prove Theorem~\ref{th:approx:lp}.  Define:
{\small
\begin{align}
  \extend{U}(r,i) \defeq & \frac{\bm 1_{i \in P_r}}{\sqrt{|P_r|}}
& \extend{V}(s,j) \defeq & \frac{\bm 1_{j \in Q_s}}{\sqrt{|Q_s|}} \label{eq:def:u:v}
\end{align}
}
where $\bm 1_\pi$ is the indicator function of a predicate $\pi$,
equal 1 when $\pi$ is true, and equal 0 otherwise.  $\extend{U}$ and
$\extend{V}$ represent mappings between the original LP and the
reduced LP, and we will show that they satisfy a relaxed version of
Eq.~\eqref{eq:def:u:v:grohe}.  Observe that the last row and last
column of both $\extend{U}$ and $\extend{V}$ are $0,0,\ldots,0,1$, and
denote by $U, V$ (without boldface) the matrices obtained by removing
the last row and last column.  Then $\hat A = U A V^T$, $\hat b = Ub$,
and $\hat c^T = c^T V^T$ (see their definitions in
Eq.~\eqref{eq:def:hat:a:b:c}).  Next, we define the following matrices
$\extend{D}, \extend{E}$, which capture the error introduced by the
mapping from $A, b, c$ to $\hat A, \hat b, \hat c$.  We also show them
as block matrices, by exposing the last row and last column:
{\footnotesize
\begin{align}
  \extend{D}\defeq \extend{A}\extend{V}^T-\extend{U}^T \hat{\extend{A}} = &&\extend{E} \defeq \extend{U}\extend{A} - \hat{\extend{A}} \extend{V} = \nonumber \\
                    \left(
                    \begin{array}{    c    |c}
                      & d_1:= \\
                                        & b \\
                      D:=AV^T-U^T\hat A & - \\
                                        & U^T\hat b \\ \hline
                      c^TV^T-\hat c^T= 0 & 0
                    \end{array}
                          \right)
                   &&   \left(
                     \begin{array}{    c    |c}
                                        & Ub \\
                                         & - \\
                        E:= UA - \hat A V & \hat b \\
                                          & =0 \\ \hline
                        e_2^T:=c^T - \hat c^T V & 0
                      \end{array}
                      \right)
 \label{eq:d:and:e}
\end{align}
}
We show that the error matrices are small:
\begin{align}
  |\extend{D}(s,i)| \leq & \frac{q}{\sqrt{|Q_s|}} & |\extend{E}(r,j)|\leq \frac{q}{\sqrt{|P_r|}}\label{eq:d:e:small}
\end{align}
We only show the first inequality, the second is identical:
{ \footnotesize
\begin{align*}
  \extend{D}(s,i)= & \sum_j \frac{\extend{A}(i,j)\bm 1_{j \in  Q_s}}{\sqrt{|Q_s|}}
                     -\sum_r \frac{\bm 1_{i \in P_r}\hat{\extend{A}}(r,s)}{\sqrt{|P_r|}}\\
  = & \frac{\extend{A}(i,Q_s)}{\sqrt{|Q_s|}}-\frac{\extend{A}(P_r,Q_s)}{P_r\cdot \sqrt{|Q_s|}}
      = \frac{1}{\sqrt{|Q_s|}}\left(\extend{A}(i,Q_s)-\frac{\extend{A}(P_r,Q_s)}{|P_r|}\right)
\end{align*}
}
where in the last line $r$ is the unique color that contains $i$.  The
quantity $\extend{A}(i,Q_s)$ represents the total weight from $i$ to
the color $Q_s$, while $\extend{A}(P_r,Q_s)/|P_r|$ is the average of
this quantity over all $i \in P_r$.  Since the coloring is $q$-quasi
stable, this difference is bounded by $q$.

For any feasible solution $x$ to the LP~\eqref{eq:sys:1}, define
$\hat x \defeq Vx$.  Since $Ax \leq b$, we derive $UAx \leq Ub$, which
becomes $(\hat AV + E)x \leq \hat b$, or $\hat AVx \leq \hat b - Ex$.
Moreover, $c^Tx = (\hat c^T V + e_2^T) x = \hat c^T\hat x + e_2^Tx$.
It follows that
$\opt(A,b,c) \leq \opt(\hat A, \hat b - Ex, \hat c) + e_2^T x$.  We
set $x:=x^*$ (an optimal solution to the LP) and observe that
Eq.~\eqref{eq:d:e:small} implies $||Ex^*||_1 \leq q ||x^*||_1$, and
$||e_2^Tx^*||_1 \leq q ||x^*||_1$.  Therefore,
Lemma~\ref{lemma:lp:lipschitz} implies
$\opt(\hat A, \hat b - Ex, \hat c) \leq \opt(\hat A, \hat b, \hat c) +
q \Delta$, for some constant $\Delta$.  We have proven that
$\opt(A,b,c) \leq \opt(\hat A, \hat b, \hat c) + O(q)$.

Conversely, for $\hat x$ any feasible solution to~\eqref{eq:sys:2},
define $x \defeq V^T \hat x$.  Since $\hat A \hat x \leq \hat b$, we
derive $U^T \hat A \hat x \leq U^T \hat b$, or
$(A V^T - D) \hat x \leq b - d_1$, which we rearrange as
$A x \leq b + (D \hat x - d_1)$.  Moreover,
$c^T x = c^T V^T \hat x = \hat c^T \hat x$.  It follows that
$\opt(\hat A, \hat b, \hat c) \leq \opt(A, b+ (D \hat x - d_1), c)$.
We set $\hat x := \hat x^*$ (an optimal solution to the reduced LP),
and observe that $||D \hat x^* - d_1||_1 = O(q)$.  Therefore,
Lemma~\ref{lemma:lp:lipschitz} implies
$\opt(A, b+ (D \hat x - d_1), c) \leq \opt(A,b,c) +q \Delta$, for some
constant $\Delta$.  This completes the proof of the theorem.

\paragraph*{Proof of Lemma~\ref{lemma:uniform}}

\label{app:proof:lemma:uniform}

We prove here Lemma~\ref{lemma:uniform}.  Extend $G$ to a network by
adding two nodes $s, t$ and setting $c(s,x) \defeq F/|X|$ and
$c(y,t) \defeq F/|Y|$ for all $x \in X, y \in Y$.  We claim that this
network admits a flow of value $F$.  The claim implies that the flow
is uniform, since all edges $(s,x)$ and $(y,t)$ must have a flow up to
their capacity.  To prove the claim it suffices to show that every cut
in the network has a capacity $\geq F$.  Let $C$ be any cut.  Define
the sets:
{\small
\begin{align*}
  S \defeq & \setof{x}{x \in X, (s,x)\not\in C}  & T \defeq & \setof{y}{y \in Y, (y,t)\not\in C}
\end{align*}
}
The cut must contain all edges $(x,y)$ with $x \in S$, $y \in T$,
hence its capacity is:
{\footnotesize
  \begin{align*}
    c(S,T) + \left(|X|-|S|\right)\frac{F}{|X|} &+ \left(|Y|-|T|\right)\frac{F}{|Y|} \\
 =  \left(c(S,T)+F - |S|\frac{F}{|X|} - |T|\frac{F}{|Y|}\right)+F
&\geq \left(c(S,T) + F - a\cdot |S|-b\cdot |T|\right) + F \geq F. \quad \square
  \end{align*}
}
\paragraph*{Proof of Theorem~\ref{th:centrality}}

\label{app:proof:th:centrality}

We write $d(a,b)$ for the length of the shortest path from $a$ to $b$
in the graph.  We claim that, for all numbers $M \geq 0$ and
$d \geq 0$, the following formula $\Phi_{M,d}$ can be expressed in
$C^3$:
{\footnotesize
\begin{align*}
  \Phi_{M,d}(s,v) \defeq & (d(s,v)=d) \wedge (\sigma(s,v)\geq M)
\end{align*}
}
The claim implies the theorem, because we can write $g(v)$ as follows.
Notice that $\sigma(s,t \mid v) = \sigma(s,v)\sigma(v,t)$ when
$d(s,t)=d(s,v)+d(v,t)$ and $\sigma(s,t \mid v) =0$ otherwise.  Then:
{\footnotesize
  \begin{align*}
    g(v) = & \sum_{M_1,M_2,M,d_1,d_2,s,t} \left\{\frac{M_1M_2}{M}\mid \Psi_{M_1,d_1}(s,v)\wedge \Psi_{M_2,d_2}(v,t) \wedge \Psi_{M,d_1+d_2}(s,t)\right\}
  \end{align*}
}
where $\Psi_{M,d} \defeq \Phi_{M,d} \wedge \neg \Phi_{M+1,d}$ asserts
that $\sigma = M$.  The claim implies that, if $u,v$ have the same
2-WL color, then, by Theorem~\ref{th:cai:fuhrer:immerman}, we have
$\Phi_{M,d}(s,u)\equiv \Phi_{M,d}(s,v)$ for all $M, d, s$, and
similarly $\Phi_{M,d}(u,t)\equiv \Phi_{M,d}(v,t)$, which implies
$g(v)=g(u)$.  It remains to prove the claim.  Recall that, for all
$d \geq 0$, the formula $\Pi_{\leq d}(x,y)$ saying ``there exists a
path of length $\leq d$ from $x$ to $y$'' is expressible in $C^3$.
For example,
$\Pi_{\leq 4}(x,y) = \exists z(E(x,z) \wedge \exists x(E(z,x) \wedge
\exists z(E(x,z) \wedge E(z,y))))$.  Then
$\Pi_{=d}(x,y) \defeq \Pi_{\leq d}(x,y)\wedge \neg \Pi_{\leq  (d-1)}(x,y)$ 
asserts that $d(x,y) =d$.  Next, we notice that:
{\footnotesize
\begin{align}
  \sigma(s,v) = \sum_{w: E(w,v) \wedge (d(s,v)=d(s,w)+1)}\sigma(s,w) \label{eq:sigma:recursive}
\end{align}
}
If $n$ is the number of nodes in the graph, then for each
$\ell = 1, n$ we denote by $\mathcal{M}_\ell$ the set of all strictly
increasing $\ell$-tuples of natural numbers
$\bm M=(M_1, \ldots, M_\ell)$, where
$0<M_1 < M_2 < \cdots < M_\ell \leq M$.  Similarly, denote by
$\mathcal{C}_\ell$ the set of all $\ell$-tuples of natural numbers
$\bm c = (c_1, \ldots, c_\ell)$, where $0 < c_i \leq M$.  Then:
{\footnotesize
  \begin{align*}
    \Phi_{M,d}(s,w) = & \bigvee_{\begin{array}{c}\bm M \in \mathcal{M}; \bm c \in \mathcal{C}\\ \sum_i c_i M_i \geq M\end{array}} \bigwedge_{i=1,\ell} \exists^{\geq c_i}w: \Pi_{=(d-1)}(s,w)\wedge \Psi_{M_i,d-1}(s,w)
  \end{align*}
}
In other words, for every combination of numbers $M_i, c_i$ such that
$\sum_i c_i M_i \geq M$, the formula checks if, for each $i$, there
exists at least $c_i$ parents $w$ at distance $d-1$ from $s$, where
$\sigma(s,w) = M_i$.  We prove that the formula is correct.  Suppose
the RHS is true.  Then, for each $i$, there are at least $c_i$
distinct nodes $w$ that satisfy the formula
$\Pi_{=(d-1)}(s,w)\wedge \Psi_{M_i,d-1}(s,w)$.  For each such $w$, we
have $\sigma(s,w) = M_i$, and therefore their contribution to the sum
in~\eqref{eq:sigma:recursive} is $c_i M_i$.  Since the numbers
$M_1, \ldots, M_\ell$ are distinct, it follows that the set of nodes
$w$ associated to distinct values $i$ are disjoint, hence their total
contribution to~\eqref{eq:sigma:recursive} is $\sum_i c_i M_i$, which
is $\geq M$ as required.

\bibliographystyle{ACM-Reference-Format}
\balance
\bibliography{bibliography}

\end{document}